\documentclass[aps,prx,reprint,twocolumn,superscriptaddress]{revtex4-2}

\usepackage{dcolumn,array}
\usepackage{bm}
\usepackage{subfigure}
\usepackage{amsmath,amsfonts,amssymb,graphics,graphicx,epsfig,color,times,natbib}
\usepackage{tikz}
\usepackage{pgfplots}
\usepackage{verbatim}
\usepackage{url}
\usepackage{diagbox}
\usepackage{slashbox}
\usepackage{makecell}

\usepackage[unicode=true,
 bookmarks=true,bookmarksnumbered=false,bookmarksopen=false,
 breaklinks=false,pdfborder={0 0 1},backref=false,colorlinks=true]
 {hyperref}
\hypersetup{
 linkcolor=magenta, urlcolor=blue, citecolor=blue, pdfstartview={FitH}, hyperfootnotes=false, unicode=true}

\newcommand{\ipic}[3][-0.5]{\raisebox{#1\height}{\scalebox{#3}{\includegraphics{#2}}}}
\newcommand{\globalH}{K}
\newcommand{\SD}{\downarrow}
\newcommand{\SU}{\uparrow}

\newcommand{\ZZ}{\mathbb{Z}}

\date{\today}

\newcommand{\var}{\text{Var}}
\usepackage{mathtools}
        \newcommand{\uPEPS}{\Psi}
\usepackage{amsthm}
\newtheorem{theorem}{Theorem}

\newtheorem{observation}{Observation}
\newtheorem{definition}{Definition}
\newtheorem{lemma}{Lemma}

\begin{document}

\title{Theory on variational high-dimensional tensor networks}

\author{Zidu Liu}\thanks{These authors contributed equally to this work.}

\author{Qi Ye}\thanks{These authors contributed equally to this work.}
 \affiliation{Center for Quantum Information, IIIS, Tsinghua University, Beijing 100084, People's Republic of China}
 
\author{Li-Wei Yu}
 \affiliation{Theoretical Physics Division, Chern Institute of Mathematics and LPMC, Nankai University, Tianjin 300071, P. R. China}
 \affiliation{Center for Quantum Information, IIIS, Tsinghua University, Beijing 100084, People's Republic of China}
 
 \author{L.-M. Duan}\email{lmduan@tsinghua.edu.cn}
\affiliation{Center for Quantum Information, IIIS, Tsinghua University, Beijing 100084, People's Republic of China}
\affiliation{Hefei National Laboratory, Hefei 230088, People's Republic of China}
\author{Dong-Ling Deng}
\email{dldeng@tsinghua.edu.cn}
 \affiliation{Center for Quantum Information, IIIS, Tsinghua University, Beijing 100084, People's Republic of China}
\affiliation{Shanghai Qi Zhi Institute, 41th Floor, AI Tower, No. 701 Yunjin Road, Xuhui District, Shanghai 200232, China}
\affiliation{Hefei National Laboratory, Hefei 230088, People's Republic of China}

\begin{abstract}
Tensor network methods are powerful tools for studying quantum many-body systems. In this paper, we investigate the emergent statistical properties of random high-dimensional tensor-network states and the trainability of variational tensor networks. We utilize diagrammatic methods and map our problems to the calculations of different partition functions for  high-dimensional Ising models with special structures. To address the notorious difficulty in cracking these models, we develop a combinatorial method based on solving the ``puzzle of polyominoes". With this method, we are able to rigorously study statistical properties of the high dimensional random tensor networks. We prove: (a)
the entanglement entropy approaches the maximal volume law, except for a small probability that is bounded by an inverse polynomial of the bond dimension; 
(b) the typicality occurs for the expectation value of a local observable when the bond dimension increases. In addition, we investigate the barren plateaus (i.e., exponentially vanishing gradients) for the high-dimensional tensor network models. We prove that such variational models suffer from barren plateaus  for global loss functions, rendering their training processes inefficient in general. Whereas, for local loss functions, we prove that the gradient is independent of the system size (thus no barren plateau occurs), but decays exponentially with the distance between the region where the local observable acts and the site that hosts the derivative parameter. Our results uncover in a rigorous fashion some fundamental properties for variational high-dimensional tensor networks, which paves a way for their future theoretical studies and practical applications. 
\end{abstract}
\date{\today}
\maketitle
\section{Introduction}\label{sec1}

 Tensor networks are promising candidates for solving quantum many-body problems~\cite{verstraete2008Matrix,Cirac2020Matrix, Eisert2010Colloquium}. Their one-dimensional version-the matrix product states (MPS)-have been proven to be efficient ansatzes in representing the ground state wave-functions of local gapped Hamiltonians in one dimension \cite{Hastings2007Area}. MPS enjoys powerful optimization algorithms including density matrix renormalization group \cite{White1992Density}, time evolving block decimation \cite{Vidal2003Efficient}, and the time-dependent variational principle \cite{Haegeman2016unifying}. Tensor networks for systems in higher dimensions, such as the projected entangled-pair states (PEPS) \cite{Verstraete2004Renormalization}, also play a crucial role in solving quantum many-body problems. However, contracting a generic PEPS is proven to be a $\#$P-hard problem \cite{schuch2007computational}. This makes the computation of expectation values of physical observables, which are crucial in solving quantum many-body problems, unattainable on classical computers for large system sizes. Despite the challenge in optimizing high-dimensional tensor networks, many approximation numerical methods have been proposed~\cite{Maeshima2001Vertical,Verstraete2004Renormalization,Levin2007tensor,jordan2008classical,jiang2008accurate,xie2009second,xie2012coarse,corboz2016variational,Vanderstraeten2016gradient,fishman2018faster,corboz2013tensor} and intriguing progresses have been made along this direction. In addition, high-dimensional tensor networks are also powerful tools in building profound connections between the anti-de Sitter/conformal field theory (AdS/CFT) correspondence in high-energy physics~\cite{maldacena2003eternal,ryu2006holographic,ryu2006aspects}
  and the concept of quantum error correction~\cite{Nielsen2010Quantum} in the context of quantum computing. In particular, the entanglement structure of the CFT in the AdS/CFT correspondence can be mapped onto a tensor network representation, which can in turn be used to construct quantum error-correcting codes~\cite{pastawski2015holographic, zhao2016Bidirectional, hayden2016holographic, cao2018bulk, donnelly2017living}. Recently, high-dimensional tensor networks have also been used in the machine learning domain to model complex data structures, such as data classification and generating~\cite{liu2023tensor, cheng2021supervised, vieijra2022generative}. By constructing a tensor network that captures the underlying structure of the data, it is possible to perform efficient computations in these tasks.

 Besides of the successful applications of tensor network methods, the studies of random tensor network states offer a promising avenue for the understanding of quantum statistical mechanics. The random matrix product state, a one-dimensional version of the random tensor network state, has been extensively studied to explore the statistical properties of many-body systems. For instance, it has been used to investigate the typicality of global and local observables~\cite{Silvano2010Statistical, Silvano2010Typicality,Haferkamp2021Emergent}, the decay of the correlations~\cite{svetlichnyy2022matrix,lancien2022correlation}, entanglement entropy~\cite{collins2012matrix,Haferkamp2021Emergent,gonzalez2018spectral}, and nonstabilizerness~\cite{haug2023quantifying, chen2022magic}. Random matrix product states can also be treated as ground states of disorder parent Hamiltonians~\cite{Movassagh2017Generic, jauslin2022random, lemm2019gaplessness} and typical representations of one-dimensional quantum phases of matter~\cite{schuch2011classifying}. Recently, extensive attentions have been attracted by the barren plateau problem, which refers to the vanishing of gradients in training variational ansatzes~\cite{Mcclean2018Barren,Cerezo2021Higher,Sharma2020Trainability, Arrasmith2020Effect,Zhao2021Analyzing, Pesah2020Absence,Marrero2020Entanglement,Patti2020Entanglement,Holmes2021Barren,Uvarov2021Barren, martin2022barren}. Such a problem for machine learning models based on random matrix product states has been systematically studied~\cite{liu2022presence, Garcia2023barren}. 
 However, extending these analyses to high-dimensional cases is notably more challenging, as the problems related to random high-dimensional tensor network states involve solving statistical models that are highly intricate.  Despite these challenges, the investigation of high-dimensional random tensor network states holds great promise for the applications of tensor network models as a variational model and for advancing the understanding of random quantum many-body systems.

In this paper, we study the emergent statistical properties of high-dimensional random tensor-network states and the trainability of such variational ansatzes. In particular, we focus on the random tensor-network states in an isometric form, which is a generalization of the canonical form for the matrix product states to higher dimensions \cite{Zaletel2020isometric}. Such models feature an advantage of enabling efficient numerical methods for contraction~\cite{lin2021efficient, tomohiro2020isometric,tepaske2021three,Zaletel2020isometric,Haghshenas2022variational} and are related to a set of tensor-network states that can be experimentally prepared in an efficient way~\cite{wei2022generation,wei2022sequential}. We rigorously analyze the statistical properties of such random high-dimensional tensor-network states by mapping the unitary integration into calculating the partition functions of classical Ising models in high dimensions. To handle the notorious difficulty in computing such partition functions, we propose a powerful excited-spin-string (ESS) method, which is derived from counting the number of polyominoes in a high-dimensional lattice \cite{golomb1996polyominoes}. With this method, we study the statistical properties of the random tensor-network states in a rigorous fashion, including their normalization factor, entanglement entropy, and typicality. In addition, we explore the landscapes of different types of loss functions for variational models based on tensor networks. For global loss functions, we prove that they suffer from barren plateaus in general. Whereas, for local loss functions, we prove that there is no barren plateau and the gradient decays exponentially with the distance between the region where the local observable acts and the site that hosts the derivative parameters.    
We also perform numerical simulations for the general 2D tensor networks to support our theoretical results. Our results establish a generic theory for variational high-dimensional tensor networks in a rigorous fashion, which would provide a valuable guide for both practical applications and future theoretical studies.


\section{Notation and Preliminary}
\label{Notation_und_Preliminary}
\begin{figure}
    \centering
    \includegraphics[scale=0.25]{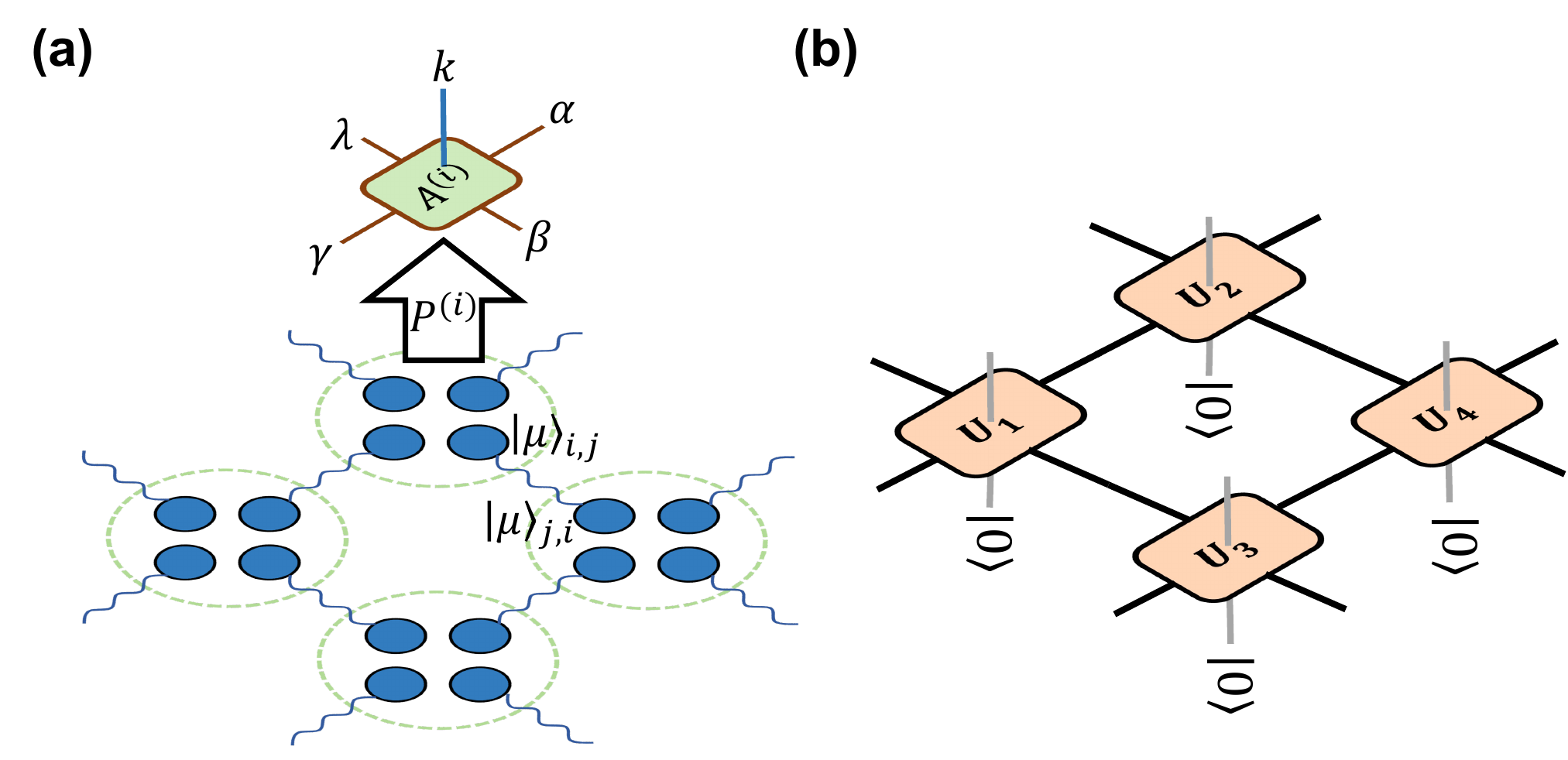}
    \caption{{Illustration of the high-dimensional random tensor-network states.} (a) Construction of a projected entangled-pair state on a 2D lattice.  The blue circles indicate virtual qubits, which are connected by solid wires to form a maximally entangled state with virtual bond dimension $D$. A projector $P$ map the virtual qubits into a physical qubit with dimension $d$. (b) Illustration of the unitarily embedded high-dimensional random tensor-network states, where each $D^2d \times D^2d$ unitary $U_i$ is randomly draw from Haar measure and applied to the reference state $|0\rangle$.}
    \label{illustration_peps}
    \end{figure}
 We consider an $L_1\times L_2$ rectangular lattice with the vertices $V = (1,2,\dots,  L_1\times L_2)$ and a set of edges $E$  connecting these vertices. As shown in Fig.~\ref{illustration_peps} (a), to construct a PEPS  $|\Psi\rangle$, we first assign each vertex $i$ with several auxiliary spins (the blue circles), where each auxiliary spin $|\mu\rangle_{i,j}$ in the $i$-th site is connected with its neighboring spin $|\mu\rangle_{j,i}$ in a maximal entanglement state by a virtual bond $e$. More explicitly, the state of these ancillary spins is given by $|\phi \rangle = \bigotimes_{e\in E} \sum_{\mu=1}^{D} |\mu\rangle_{i,j} \otimes |\mu\rangle_{j,i}$, where $D$ denotes the virtual bond dimension. For each site $i$, we assign a linear map $P^{(i)} = \sum_{\alpha, \beta, \gamma, \lambda = 1}^{D}\sum_{j=1}^{d}A_{\alpha, \beta, \gamma, \lambda, j}^{(i)}|j\rangle \langle \alpha, \beta, \gamma, \lambda |$, which maps the entangled pairs of these ancillary spins into the physical spins with the physical dimension denoted by $d$.  With this framework, the PEPS reads: $|\Psi\rangle = \bigotimes_{i} P^{(i)} |\phi\rangle $.

In our setup, each projector $P^{(i)}$ can be unitarily embedded into a $D^2 d \times D^2 d$ unitary matrix $U_i$ [as shown in the Fig. ~\ref{illustration_peps} (b)]. One can obtain the $D^4 \times d$ local tensor $A_{\alpha, \beta, \gamma, \lambda, j}^{(i)}$ by contracting with a constant reference state $|0\rangle$ on a physical leg. These local tensors are drawn uniformly from the Stiefel manifold of isometries and approximately form the unitary $t$-design. 

The integral over the $t$-th moment of a Haar-random unitary is given by the Weingarten calculus~\cite{brouwer1996diagrammatic, collins2006integration}. In the following, we utilize diagrammatic languages to describe the Haar-random integration over each local tensor, which maps the integration results into a spin representation.
More specifically, for the first moment of a Haar-random unitary $\overline{U\otimes U^\dagger}$, we have:
\begin{equation}
      \ipic{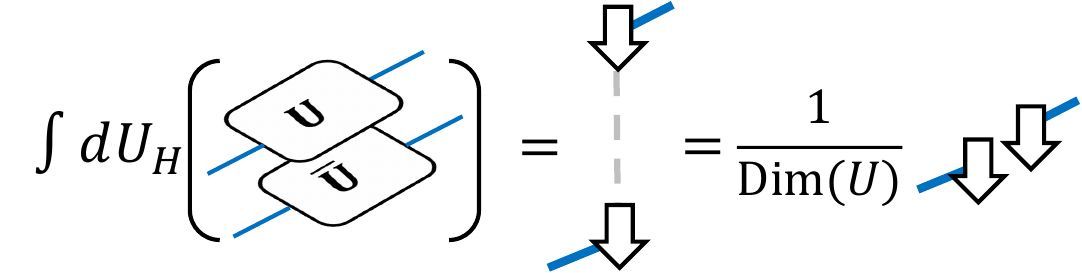}{0.3},
\end{equation}
where $\text{Dim}(U)$ indicates the corresponding dimension of the unitary tensor, $\bar{U}$ is the complex conjudgate of $U$, and the spin variable takes the connection:
\begin{equation}
    \ipic{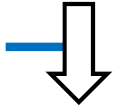}{0.3} = \ipic{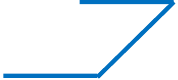}{0.3}.
\end{equation}
The weight factor $1/\text{Dim}(U)$ is given by the Weingarten calculus. Indeed, each spin can only take one value (we denote as a thick down arrow), which corresponds to a spin-$0$ model. The second moment integral over Haar measure $\overline{U\otimes U^\dagger \otimes U \otimes U^\dagger}$ has the following spin representation:
\begin{equation}
          \ipic{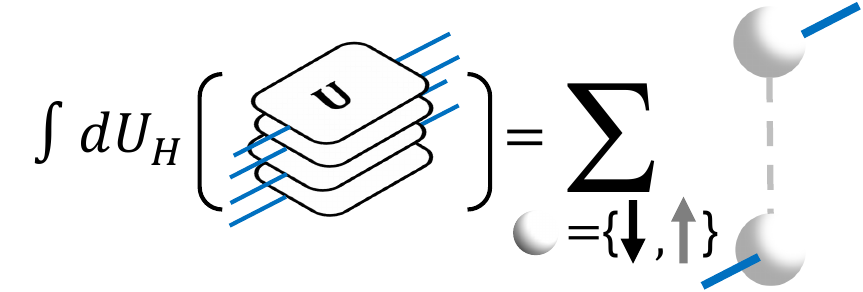}{0.3},
\end{equation}
and we denote the spin states $\uparrow/\downarrow$ as:
\begin{equation}
    \ipic{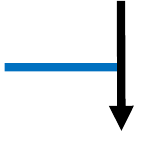}{0.3} = \ipic{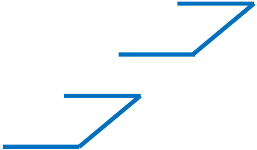}{0.3} , \qquad
    \ipic{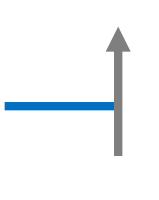}{0.3} = \ipic{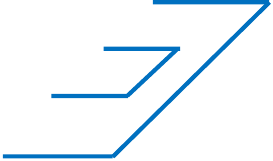}{0.3},
\end{equation}
with different virtual weight given by the Weingarten function:
\begin{align}
    \ipic{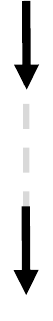}{0.25} &=\ipic{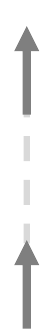}{0.25} = \frac{1}{\text{Dim}(U)^2- 1},\\
    \ipic{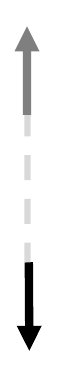}{0.25} &=\ipic{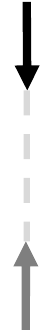}{0.25} = -\frac{1}{[\text{Dim}(U)^2 - 1][\text{Dim}(U)]}.
\end{align}
Here the Haar integral of each local unitary is mapped to a spin-$1/2$ model. 

We can also apply this kind of representation on the random tensor network states. For each local unitary tensor the second moment integral is:
\begin{equation}
\begin{aligned}
    \label{eq_twodesign}
     \ipic{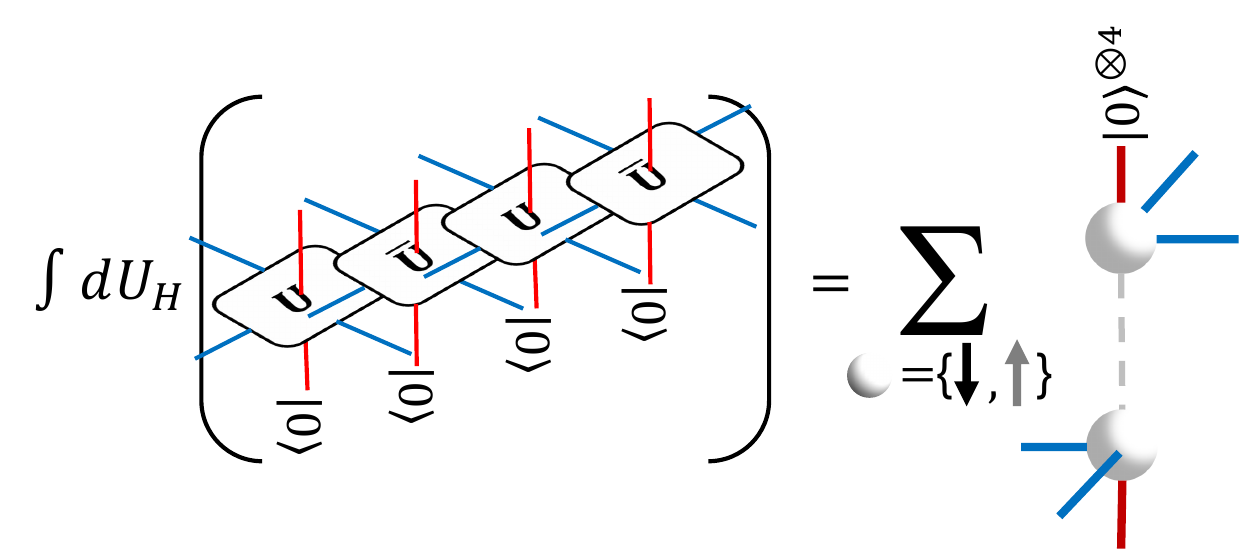}{0.3},
    \end{aligned}
\end{equation} 
where we define the fourfold reference state as $|\vec{0}\rangle = |0\rangle^{\otimes 4}$. 
 We notice that the contraction between both spin up and spin down state and the fourfold state $|0\rangle^{\otimes 4}$ have $\langle \uparrow/\downarrow||0\rangle^{\otimes 4} = 1$. Hence, we can ignore the reference state in the following calculations. 
When we calculate the second moment integral over Haar measure for a random tensor network with a size of $L \times L$, denoted by $\int dU_H \left(|\Psi\rangle\langle \Psi|\right)^{\otimes 2}$, it gives an Ising-like spin model. The specific manner in which the dangling legs (i.e., those not connected to any other vertex) are connected will be determined later for particular problems.

 In the above representation, the calculation of such integral is equivalent to solve the partition function of the spin models:
 \begin{equation}
     Z = \sum_{\vec{\sigma}} e^{-H({\vec{\sigma}, h_z})},
 \end{equation}
where the external field $h_z$ depends on the tensors connected to the dangling physical dimension, $\vec{\sigma}$ denotes the spin configuration, $H$ corresponds to the energy of a specific configuration $\vec{\sigma}$.

\section{Statistical properties of the random high-dimensional tensor-network states}
\label{Statistical_properties}
\begin{figure*}
\centering
\includegraphics[scale=0.25]{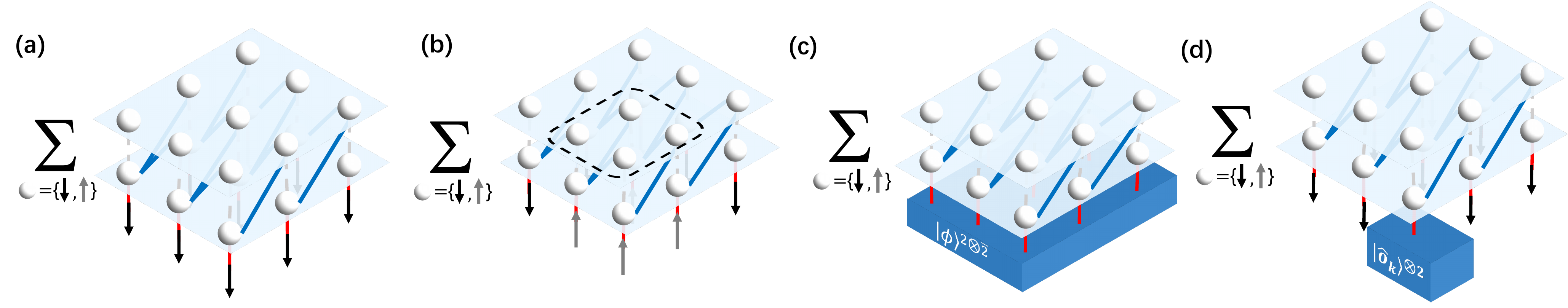}
\caption{{Spin representation of the Haar random integration over the random tensor-network states.}  (a)  A graphic representation of the second moment integral of the overlap function $\overline{\langle \Psi |\Psi\rangle^2} = \int dU_H \langle \psi | \psi \rangle^2$ over Haar measure $dU_H$, where each dangling physical leg is connected to a spin down state. This corresponds to a classical Ising model with uniform external fields. (b) A graphic representation of the averaged purity $\overline{\text{Tr}(\rho_A^2)}$ of a subsystem $A$ (the dotted region).  (c) A graphic representation of the integration over the second moment global loss function $\overline{|\langle \Psi|\phi\rangle \langle \phi |\Psi\rangle|^2}$, where the dangling physical legs are connected to a fourfold global state in a choi-Jamiolkowski isomorphism $|\phi\rangle^{2\otimes \bar{2}} = (|\phi\rangle\langle \phi |)^{\otimes 2}$ \cite{jamiolkowski1972linear,choi1975completely}. (d) A graphic representation of $\overline{|\langle \Psi|\hat{O}_k|\Psi\rangle|^2}$, where we denote $|\hat{O}_k\rangle^{\otimes2} = \hat{O}_k\otimes \hat{O}_k$. }
\label{illustration_Methods}
\end{figure*}

\subsection{Fluctuation of the norm}
\label{norm_fluctuation_title}
We start by discussing the computation of the normalization factor and the fluctuation of the random high-dimensional tensor-network states. In our setting, the wave function of such a tensor-network state is not necessarily normalized. Using the concentration of measure for the unitary group~\cite{hayden2004randomizing,Gross2009Most,Bremner2009random}, we prove that the norm of each state vector $|\Psi \rangle$ is exponentially concentrated around unity. This part of the proof not only demonstrates that the influence of the norm factor is negligible in subsequent calculations when the system is large, but also serves as a good example to systematically introduce our technique framework based on enumerating polyominoes.

\begin{theorem}\label{theorem:normalized} 
Consider a random tensor-network state $|\Psi\rangle$ with  bond dimension $D\leq 2$ and  physical dimension $d\leq 2$ in a square lattice with volume $V = L \times L$, where $L \geq 2$. The norm of the random tensor-network state is exponentially concentrated around one as $L$ increases:
\begin{equation}
    \text{Pr}\left(|\langle \Psi |\Psi \rangle-1|\geq\epsilon\right)\leq O\left(\frac{\kappa^{L}}{\epsilon^2}\right),
\end{equation}
with $\kappa = 0.7$ for a $2$D random tensor-network state. 
\end{theorem}

Assuming that each local unitary forms a $1$-design, it is straightforward to verify that the expectation value of the norm for a random tensor-network state, denoted as $\overline{\langle \Psi|\Psi \rangle}$, is equal to one, i.e., $\overline{\langle \Psi|\Psi \rangle } = \int dU_{H}\langle \Psi(\Theta)|\Psi(\Theta) \rangle = 1$. Therefore, in order to apply the Chebyshev inequality, one only needs to compute the variance of the norm, which is given by $\text{Var}(\langle \Psi | \Psi \rangle ) = \overline{\langle \Psi | \Psi \rangle^2} - \overline{\langle \Psi | \Psi \rangle}^2 =\overline{\langle \Psi | \Psi \rangle^2} -1$. Under the assumption that each local unitary forms a 2-design, we use the Haar-random integration to calculate the variance of the norm. Our proof consists of two parts. First, we map the computation of the Haar-random integration $\overline{\langle \Psi|\Psi \rangle^2}$ to a statistical model, specifically the two-dimensional classical Ising model with external fields. Second, we develop a combinatorial method based on enumerating the toric polyominoes to calculate the partition function of the Ising model. {It is worth noting that solving a general Ising model in high dimension is a long-standing challenge, open for decades ~\cite{mccoy2014two}}. Here, we find that, by introducing an imaginary part to the coupling coefficient in the Ising Hamiltonian, one can reduce the complexity of computing the partition function of the high-dimensional Ising model substantially.



\subsubsection{Mapping to statistical models}

Based on Eq.~(\ref{eq_twodesign}), the calculation of the variance can be mapped to a two-layer classical Ising model with uniform external fields [see Fig.~\ref{illustration_Methods} (a)]. The corresponding partition function is given by:
\begin{equation}
    Z = \frac{1}{C}\sum_{\vec{\sigma}} e^{- H(\vec{\sigma})},
\end{equation}
where the normalization constant $C = [\text{i} D^2 d/(D^4 d^2 -1)]^{V}$
and the spin configuration reads $\vec{\sigma} = \{\dots, \sigma^{(1/2)}_{x,y},\dots\}$, 
with $\sigma_{x,y}^{(1)},\sigma_{x, y}^{(2)}\in \{\uparrow, \downarrow\}$ representing the spin-up and spin-down states at the site $(x,y)$. The indices $(1/2)$ indicate the bottom and upper layers, respectively. 
The Ising Hamiltonian is given by: 
\begin{equation}
\begin{aligned}
\label{norm_ising_hamiltonian}
       H(\vec{\sigma}) =& \sum_{x,y} H_{x,y}\left(\sigma_{x,y}^{(1)}, \sigma_{x,y}^{(2)}, \sigma_{x,y+1}^{(2)},\sigma_{x+1,y}^{(2)} \right)\\
       =&\sum_{x,y} \frac{1}{2} \sigma_{x,y}^{(1)}\left[\left(\sigma_{x,y+1}^{(2)} +\sigma_{x+1,y}^{(2)}\right)J_1-\sigma_{x,y}^{(2)} J_2+h_z\right], 
\end{aligned}
\end{equation}
where the Ising coupling coefficient $J_2 = \text{i}\pi + \text{log}(D^2 d)$, $J_1 = -\text{log}(D)$, and the strength of the external field is $h_z = - \text{log}(d)$. 
In fact, owing to the absence of coupling between spins on each layer, we can compress this model into a two-dimensional hexagonal Ising model. Nevertheless, we adopt this double-layer structure to simplify the  subsequent proof.


At each site $(x,y)$, we define: 
\begin{equation}
\label{single_site_contraction}
\begin{aligned}
f(\sigma_{x,y}^{(2)}, \sigma_{x,y+1}^{(2)},\sigma_{x+1,y}^{(2)})\coloneqq \ipic{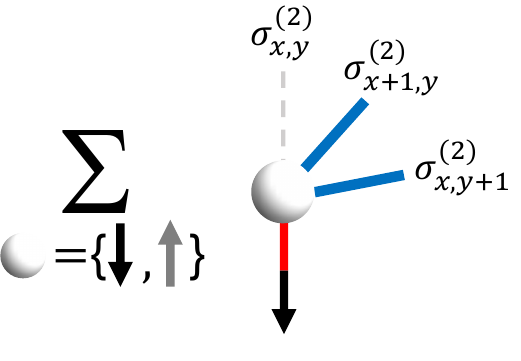}{0.4},
\end{aligned}
\end{equation}
so that 
\begin{equation}
\begin{aligned}
        Z &= \frac{1}{C}\sum_{\vec{\sigma}}e^{-H\left(\vec{\sigma}\right)} = \sum_{\vec{\sigma}}\prod_{x, y}f({\sigma}_{x,y}^{(2)}, {\sigma}_{x, y+1}^{(2)}, {\sigma}_{x+1, y}^{(2)}) \\
    &= \sum_{\vec{\sigma}} e^{-H_{\text{eff}}\left(\vec{\sigma}^{(2)}\right)},
\end{aligned}
\end{equation}
where $H_{\text{eff}}$ is an effective Hamiltonian for the spin configuration $\vec{\sigma}^{(2)}$.

When the spins in the upper layer take different states, the Eq.~(\ref{single_site_contraction}) can be expressed as:
\begin{subequations}
    \begin{align}
     f(\downarrow,\downarrow,\downarrow) &=1,\label{single_integration1}\\
    f(\downarrow,\downarrow,\uparrow) &= f(\downarrow,\uparrow,\downarrow) = \frac{D^3d^2-D}{D^4d^2 - 1}, \label{single_decay1}\\
    f(\downarrow,\uparrow,\uparrow) &= \frac{D^2d^2 -D^2}{D^4d^2 - 1}, \\
    f(\uparrow,\downarrow,\downarrow) &= 0,\label{single_integration_zero}\\
    f(\uparrow,\downarrow,\uparrow) &= f(\uparrow,\uparrow,\downarrow) =\frac{D^3d - Dd}{D^4d^2 - 1},\label{single_decay2}\\
    f(\uparrow,\uparrow,\uparrow) &= \frac{D^4d-d}{D^4d^2-1}.\label{single_decay3}
\end{align}
\end{subequations}
In the following, we denote $\vec{\sigma}$ as the spin configurations in the upper layer, where every spin in the bottom layer is summed over. Each configuration contributes a value  $\mathcal{A}[\vec{\sigma}] = \prod_{x, y}f({\sigma}_{x,y}^{(2)}, {\sigma}_{x, y+1}^{(2)}, {\sigma}_{x+1, y}^{(2)}) $ to the partition function. Following Eqs.~(\ref{single_integration1}-\ref{single_decay3}), configurations with non-zero values have the following properties: 
\begin{observation}[Properties of configurations with non-zero values]~
\label{Properties_configurations}
\begin{enumerate}
    \item If $\mathcal{A}[\vec{\sigma}]\neq 0$ and $\sigma_{x,y}=\uparrow$, one of $\sigma_{x+1, y}$ and $\sigma_{x, y+1}$ should be $\uparrow$ as well. 
    \item Each excitation in one configuration will contribute a factor that is less than one.
    \item Each pair of neighbors with different spins will contribute a factor that is less than one.
    \item The configuration with all $\sigma_{x,y}=\downarrow$ has value $1$, which corresponds to the ground state of the effective Hamiltonian $H_{\text{eff}}$. 
\end{enumerate}
\end{observation}

\begin{figure}[htp]
\centering
\includegraphics[scale=0.35]{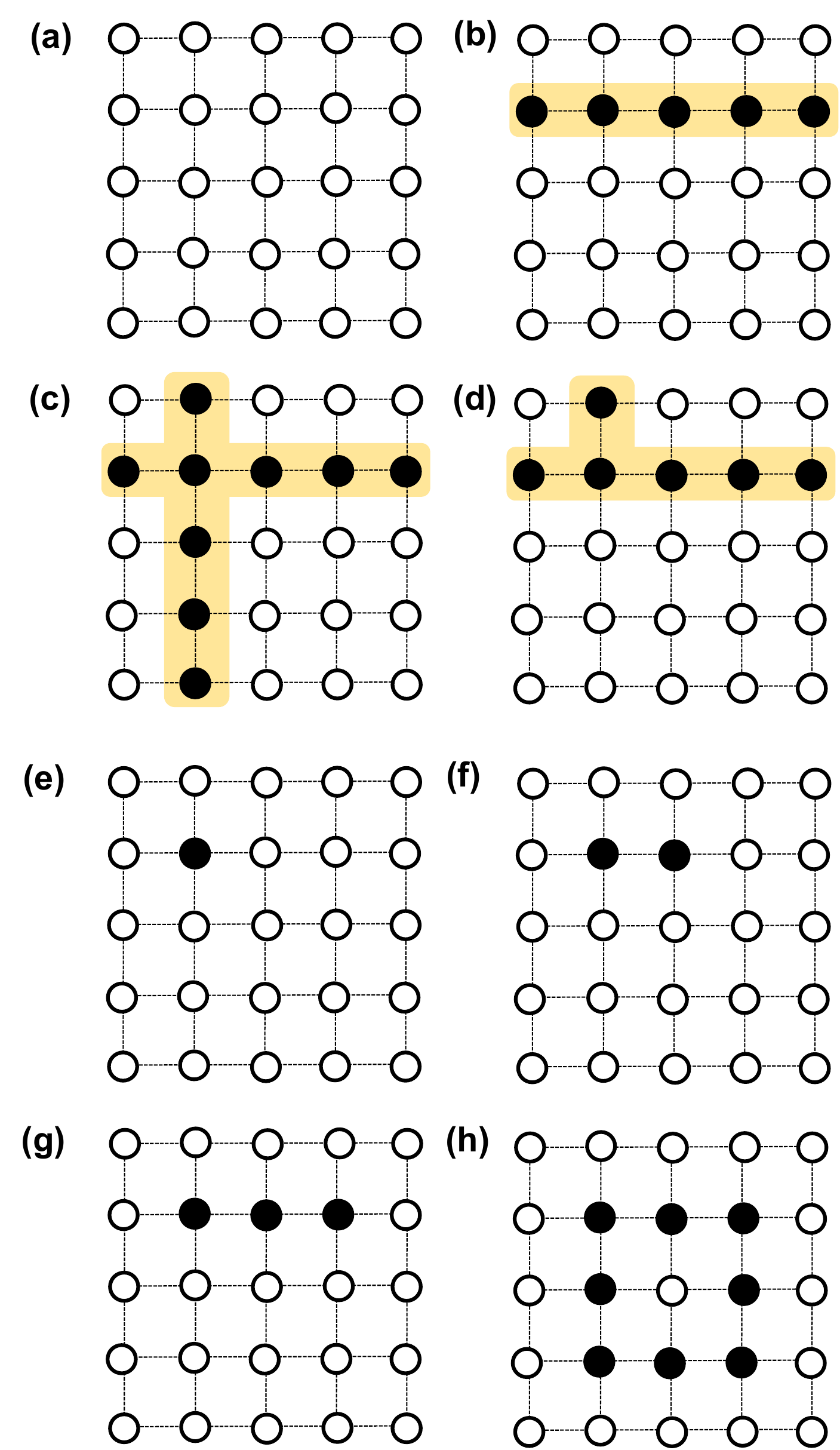}
\caption{{Illustrations of the zero and non-zero configurations in a $L \times L$ lattice.} (a) the ground state spin configuration. Here, we denote the spin-down state as a white circle, and the spin-up state as a black disc. This configuration contributes a value of $1$. (b-d) non-zero configuration with valid excited spin string (ESS), where the length of the ESS  has to be larger than $L$ and forms at least one closed loop. (e-h) an illustration of the zero configuration spin excitation. }
\label{configure_examples}
\end{figure}
The configurations with the form in Fig.~\ref{configure_examples} (b-d) contribute non-zero terms in the partition functions, where the excited spins form an excited-spin-string (ESS). Let $\Sigma$ denote the set of all spin configurations with non-zero $\mathcal{A}$-values except the trivial all-$\downarrow$ configuration. The partition function can be rewritten in term of these configurations:
\begin{equation}
\label{partition_function_excited}
    Z = 1+ \sum_{\vec{\sigma}'\in \Sigma}\mathcal{A}\left[\vec{\sigma}'\right].
\end{equation}
To calculate this partition function, we employ a method that mapping excited spin configurations to polyominoes, also referred to as animals in physical literature~\cite{golomb1996polyominoes, delest1984algebraic,redelmeier1981counting}. A polyomino is a plane geometric shape composed of identical squares joined edge to edge, as illustrated in Fig.~\ref{illustration_polyomino} (a). As shown in Fig.~\ref{configure_examples} (b-d), an ESS naturally forms a polyomino. By utilizing techniques for enumerating polyominoes, we can compute the partition function expressed in Eq.~(\ref{partition_function_excited}).  



\subsubsection{The polyominoes}

The enumeration of polyominoes is relevant to various statistical problems, including percolation problems \cite{stauffer2018introduction} and branched polymers \cite{Lubensky1979statistics}. Counting the total number of polyominoes is a challenging task, and there is no known formula or algorithm that can provide the number of polyominoes for any given size~\cite{bousquet-melou_exactly_2009}. However, some subclasses of polyominoes, such as the directed polyominoes, have exact formulas available~\cite{bousquet1998new}.


\begin{figure}[htp]
    \centering
    \includegraphics[scale=0.275]{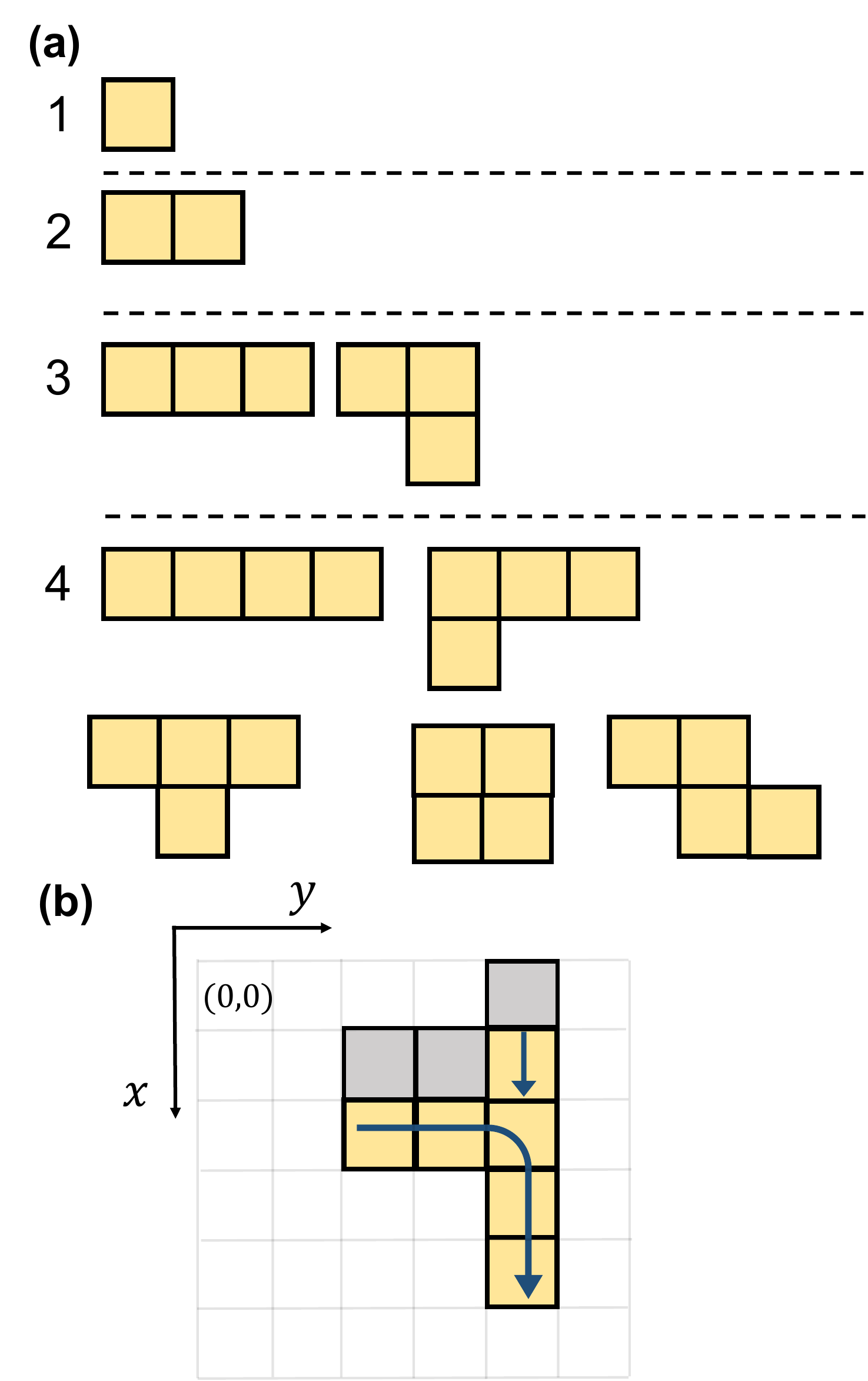}
    \caption{{Illustrations of the general polyomino and the directed polyomino} (a) Illustration of the different types of general polyominoes with the area ranging from $1$ to $4$. Here we consider the different rotations of a polyomino as the same type. (b) Illustration of a directed polyomino rooted at $(0, 0)$ with area $m = 6$, perimeter $p=14$, and upper perimeter $n=3$ (as shown in the grey square).}
    \label{illustration_polyomino}
    \end{figure}

The directed polyominoes lie in an infinite square lattice. In this context, the term ``directed" indicates a polyomino extends to the right and below, and it stops at a root. A formal definition of the directed polyominoes is presented as follows:
\begin{definition}[Directed polyominoes]
    A directed polyomino in the Euclidean plane rooted at $(x_0, y_0)$ is a finite subset $\vec{\tau}\subseteq \ZZ \times \ZZ$ ($\ZZ$ denotes the set of integers)
    such that,
    \begin{itemize}
        \item $(x_0, y_0)\in \vec{\tau}$,
        \item For $(x, y)\in \vec{\tau}$ and $(x,y)\neq (x_0, y_0)$, at least one of $(x+1, y)$ and $(x, y+1)$ are in $\vec{\tau}$ as well.
    \end{itemize}
\end{definition}
In this paper, directed polyomino is also called a plane polyomino. 
A polyomino can be described using its area, perimeter, and upper perimeter. The formal definitions of these terms are listed as follows:
\begin{definition}[Area, perimeter, and upper perimeter]\label{def_area_and_upper_site_perimeter}
    Let $\vec{\tau}$ be a directed polyomino.
    \begin{itemize}
        \item The area $m$ of $\vec{\tau}$ is the size of the polyomino, i.e., $|\vec{\tau} |$.
        \item The perimeter $p$ of $\vec{\tau}$ is the number of edges on the boundary of $\vec{\tau}$. Formally speaking, $p=|\{(x, y)\in \ZZ\times \ZZ|\tau_{x, y}\neq \tau_{x+1, y}\}|+|\{(x, y)\in \ZZ\times \ZZ|\tau_{x, y}\neq \tau_{x, y+1}\}|$.
        Here $\tau_{x, y}=1$ if $(x, y)\in \vec{\tau}$ and $\tau_{x, y}=0$ if $(x, y)\not\in \vec{\tau}$.
        \item The upper perimeter $n$ of $\vec{\tau}$ is the number of horizontal edges on the top boundary of $\vec{\tau}$. Formally speaking, $n=|\{(x, y)\in \ZZ\times \ZZ|\tau_{x, y}=0, \tau_{x+1, y}=1\}|.$
    \end{itemize}
\end{definition}

We can enumerate directed polyominoes exactly via their generating function~\cite{bousquet1998new}.
\begin{lemma}[\cite{bousquet1998new}]
    Let $D_{m,n}$ be the number of directed polyominoes rooted at $(0, 0)$ with area $m$ and upper perimeter $n$. $p$ and $q$ are two variables. We introduce the generative function of $D_{m,n}$ as  
    \begin{equation}\label{equ_D_definition}
        G(q, p)=\sum_{m,n}D_{m,n}q^m p^n.
    \end{equation}
    Then we have:
    \begin{equation}\label{equ_D}
        G(q, p) = \frac{p}{2}\left(\sqrt{\frac{(1+q)(1+q-qp)}{1-q(2+p)+q^2(1-p)}}-1\right).
    \end{equation}
\end{lemma}
This lemma implies that when $q, p\in (0, 1]$ such that $|q(2+p)+q^2(1-p)|<1$, the power series $\sum_{m,n}D_{m,n}q^mp^n$ converges to a finite value $G(q, p)$.
We will see that the partition function $Z$ has a similar power-series form, so the remaining problem reduces to determining the number of configurations in $\Sigma$ by $D_{m,n}$.

For simplicity, we consider a square lattice with a periodic boundary condition, forming an $L\times L$ toric graph. A vertex on the torus is indexed by $(x, y)\in \ZZ_L\times \ZZ_L$, where $\ZZ_L=\{0, 1, \cdots, L-1\}$. We remark that our technique can be easily generalized to rectangle lattices. 
A vertex $(x, y)$ is connected to $((x\pm 1) \% L, y)$ and $(x, (y\pm 1)\% L)$, where $\% L$ means the residue modulo $L$. All $\uparrow$ in an excited configuration $\vec{\sigma}\in \Sigma$ naturally forms a polyomino on the torus. We will call the excited configuration a toric polyomino, so as to be consistent with the definition of plane polyominoes. A polyomino may have many connected components. Each component is an excited-spin-string (ESS):
\begin{definition}[ESS]
    An ESS in the toric graph $T_L$ rooted at $(x_0, y_0)$ is a connected subset of $\ZZ_L\times \ZZ_L$, such that
    \begin{itemize}
        \item $(x_0, y_0)\in \vec{\sigma}$,
        \item For $(x, y)\in \vec{\sigma}$ and $(x,y)\neq (x_0, y_0)$, at least one of $((x+1)\% L, y), (x, (y+1)\% L)$ are in $\vec{\sigma}$ as well.
    \end{itemize}
\end{definition}
By the definitions, a configuration $\vec{\sigma}\in \Sigma$ is a union of several ESSs. Similar to plane polyominoes, we can define area $m$, upper perimeter $n$, and perimeter $p$ for every ESS and toric polyomino. From the Properties~\ref{Properties_configurations} (2-3), we can upper bound $\mathcal{A}[\vec{\sigma}]$ by the area and perimeter of toric polyominoes as shown in the following lemma. 
\begin{lemma}\label{lem_area_and_perimeter_bound}
    Let $\vec{\sigma}\in \Sigma$ be a configuration with area $m$, perimeter $p$, and upper perimeter $n$, then
    \begin{equation}
        \mathcal{A}[\vec{\sigma}]\leq q_a^m q_p^{p}\leq q_a^m q_u^{n}
    \end{equation}
    where the area parameters $q_a= f(\uparrow, \uparrow, \uparrow)=\frac{D^4d-d}{D^4d^2-1}\lesssim \frac{1}{d}$, the perimeter parameter $q_p = f(\downarrow, \downarrow, \uparrow)=\frac{D^3d^2-D}{D^4d^2-1}\lesssim \frac{1}{D}$, and the upper perimeter parameter $q_u=q_p^2\lesssim \frac{1}{D^2}$.
\end{lemma}
\begin{proof}
We note that $f(\downarrow, \uparrow, \uparrow)=\frac{D^2d^2-D^2}{D^4d^2-1}\leq q_p^2, f(\uparrow, \downarrow, \uparrow)=f(\uparrow, \uparrow, \downarrow)=\frac{D^3d-Dd}{D^4d^2-1}\leq q_aq_p$.
Hence, for all $\sigma_{1}, \sigma_2, \sigma_3\in \{\uparrow, \downarrow\}$, we have
\begin{equation}
    f(\sigma_1, \sigma_2, \sigma_3)\leq q_a^{\mathcal{I}(\sigma_1 = \uparrow)}p_q^{\mathcal{I}(\sigma_1\neq \sigma_2)+\mathcal{I}(\sigma_1 \neq \sigma_3)},
\end{equation}
where $\mathcal{I}(P)$ is the indicator function of an event $P$, i.e., $\mathcal{I}(P)=1$ if the event $P$ occurs, and 0 otherwise. According to the definition, $m=\sum_{x, y}\mathcal{I}(\sigma_{x,y}=\uparrow), p=\sum_{x, y}(\mathcal{I}(\sigma_{x, y}\neq \sigma_{x+1, y})+\mathcal{I}(\sigma_{x, y}\neq \sigma_{x, y+1}))$.
Therefore, $\mathcal{A}[\vec{\sigma}]=\prod_{x,y}f(\sigma_{x,y}, \sigma_{x+1, y},\sigma_{x, y+1})\leq q_a^mq_p^p$.
Furthermore, it is easy to see that the upper perimeter always equals to the lower perimeter, so $p\ge 2n$ and $q_a^mq_p^p\leq q_a^mq_u^n$.
\end{proof}

This lemma implies that $\mathcal{A}[\vec{\sigma}]$ exponentially decays with the area and the upper perimeter. The partition function can be upper bounded by:
\begin{equation}
\label{partition_upperbound}
    Z = 1+ \sum_{\vec{\sigma}} \mathcal{A}[\vec{\sigma}] \leq 1 + \sum_{m,n} \widetilde{D}_{m,n} q_a^{m} p_u^n,
\end{equation}
where we set the number of toric polyominoes with area $m$ and upper perimeter $n$ as $\widetilde{D}_{m,n}$. A toric polyomino is very similar to a plane polyomino except it lies on a torus with periodic boundary. We can simply embed a plane polyomino $\vec{\tau}$ to the torus by mapping each $(x, y)\in \vec{\tau}$ to $(x\% L, y\% L)\in \mathbb{Z}_L\times \mathbb{Z}_L$. The resulting shape is a toric polyomino. Therefore, the number of plane polyominoes and  that of toric polyominoes with the same area and upper perimeters are roughly the same. In Appendix~\ref{proof_theorem1}, We propose a bridge transformation to make this intuition rigorous. In short, we can write the summation in Eq.~(\ref{partition_upperbound}) as:
\begin{align}
\label{partition_upperbound_1}
    \sum_{m,n} \widetilde{D}_{m,n} q_a^{m} p_u^n &= \sum_{m\ge L,n} \widetilde{D}_{m,n} q_a^{m} p_u^n\nonumber\\
    &\leq \text{poly}(L)\sum_{m\ge L, n} D_{m,n} (q_a)^{m} (p_u)^n\nonumber\\
    &\leq \text{poly}(L)\eta^L\sum_{m\ge L, n} D_{m,n} (q_a/\eta)^{m} (p_u)^n\nonumber\\
    &\leq \text{poly}(L)\eta^L G(q_a/\eta, p_u),
\end{align}
where the area of any toric polyomino is at least $L$, and $\text{poly}(L)$ is a coefficient brought by the bridge transformation.
$\eta\in (0, 1)$ is a constant such that $G(q_a/\eta, p_u)$ converges to a finite value.
For $d, D\ge 2$, we can choose $\eta=25/26<0.97$, therefore $Z\leq 1 + O(0.97^L)$. Then, we can implement the Chebyshev's inequality~\cite{Tchebichef1867Des} to bound the fluctuation of the norm of random tensor network states. In the appendix we propose a reflection trick that can improve the bound from $0.97$ to $0.7$.
We remark that with a large enough system size, the fluctuation of the norm can be suppressed exponentially.  Theorem~\ref{theorem:normalized} indicates that we can compute other statistical properties of the random tensor network states without considering the normalization factor. 



\subsection{Entanglement entropy}
\subsubsection{Area law entanglement entropy for large subsystems}
Now we apply the above theorem to investigate the entanglement properties of the random tensor network states. The purity $\text{Tr}(\rho_A^2)$ for a subregion $A$ (with the size $|A| = l \times l$) determines the second R\'{e}nyi entropy. We can use a doubled copy of random tensor network states to obtain the second R\'{e}nyi entropy (This corresponds to the swap test approach~\cite{Ekert2002lineardensitymatrix}):
\begin{equation}
    \label{renyi_entropy}
    S_2(A) = -\text{log}\left[ \frac{\text{Tr}(X_A |\Psi\rangle \langle \Psi | \otimes |\Psi\rangle \langle \Psi |)}{\text{Tr}(|\Psi\rangle \langle \Psi | \otimes |\Psi\rangle \langle \Psi |)}\right],
\end{equation}
where $X_A$ is a swap operator that only acts on the subregion $A$ of the two copies. We note that the denominator term is equivalent to the norm of a random tensor network. Since this denominator is concentrated around $1$ as the system size is large enough (as stated in Theorem~\ref{theorem:normalized}), we will ignore this term in the later discussions. We first focus on the entanglement properties where the subsystem size $|A|$ is large. 
\begin{theorem}
    \label{theorem:entropy} 
    Given a subset qubits system $A$ with a volume $|A|=l\times l$ in the $V = L \times L$ tensor network with $L \rightarrow \infty$, we set the bond dimension $D\ge 3$, and physical dimension $d \ge 2$. Then, the averaged purity of the reduced density matrix $\rho_A$ can be bounded by the following relations:
    \begin{equation}
    \lambda_1 q_a^{|A|}+\lambda_2 q_p^{|\partial A|}\leq \overline{\text{Tr}(\rho_A^2)} \leq c(1.91q_a)^{|A|}+c(2.9q_p)^{|\partial A|},
    \end{equation}
where $c$ is a constant independent of $d, D, l$,  $\lambda_1=(1+q_a^{-1}q_p^4/2)^{(l-1)^2/2}$, and $\lambda_2=0.7(1+q_aq_p^2/2)^{4l-1}$. 

    \end{theorem}

\begin{figure}[htp]
    \centering
    \includegraphics[scale=0.35]{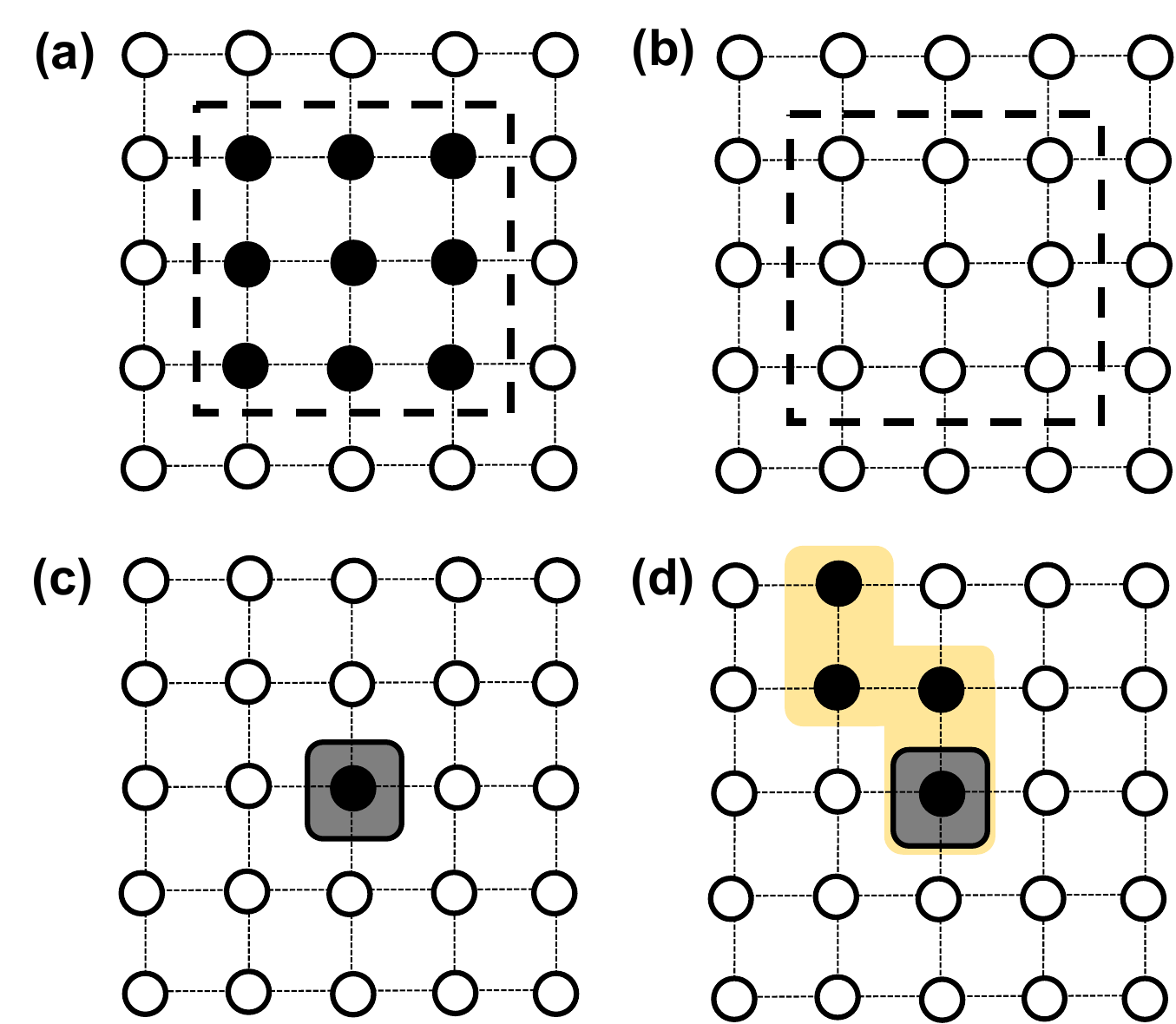}
    \caption{{Corresponding spin configurations of the averaged purity of reduced density matrix $\overline{\text{Tr}(\rho_A^2)}$ and the averaged second order expectation value $\overline{\langle \Psi|\hat{O}|\Psi \rangle^2}$.}  (a) the ground state configuration of the averaged purity as bond dimension $D < D_c$. The thick dotted line separate the two subsystems of inside and outside of region $A$. (b) the ground state configuration as bond dimension $D > D_c$. (c) the ground state configuration of $\overline{\langle \Psi|\hat{O}|\Psi \rangle^2}$, where the grey square indicates the site that hosts for the local observable $\hat{O}_k$. (d) the excited states configurations of $\overline{\langle \Psi|\hat{O}|\Psi \rangle^2}$, where the ESS has at least one endpoint on the site that hosts the local observable. }
    \label{entropy_typicality}
    \end{figure}

The averaged purity of the reduced density matrix reads: $Z_p = \overline{\text{Tr}(\rho_A^2)}$. According to Eq.~(\ref{renyi_entropy}), the external field $h_z^{(A)}$ inside the subregion $A$ is opposite to the outside region due to the swap operator $X_A$, as shown in Fig.\ref{illustration_Methods} (b).
 Hence, we can map the averaged purity to a partition function:
 \begin{equation}
    Z_p = \sum_{\vec{\sigma}}\prod_{(x, y)\not\in A} f_{x,y}^{(1)}\prod_{(x, y)\in A}f_{x,y}^{(2)}.\label{equ_Z_p}
\end{equation} 
Here, we define:
\begin{equation}
    \begin{aligned}
            f_{x,y}^{(1)} &= f(\sigma_{x,y}, \sigma_{x, y+1}, \sigma_{x+1,y});\\
            f_{x,y}^{(2)} &= f(-\sigma_{x, y}, -\sigma_{x, y+1}, -\sigma_{x+1, y}),
\end{aligned}
\end{equation}
where $-\sigma_{x, y}$ is the opposite spin direction of $\sigma_{x, y}$. 
Let $\mathcal{B}[\vec{\sigma}]$ denote the value of $\vec{\sigma}$ so that $Z_p=\sum_{\vec{\sigma}}\mathcal{B}[\vec{\sigma}]$. 

The difference between $\mathcal{A}[\vec{\sigma}]$ and $\mathcal{B}[\vec{\sigma}]$ is due to the effect of region $A$.
Intuitively, if the area for configurations is large, the effect of region $A$ is relatively insignificant, so we use the same technique in Theorem 1 to show that their contributions to $Z_p$ is negligible. Therefore, the leading term of $Z_p$ comes from small-area configurations around $A$.
There are two typical configurations best demonstrating the effect of $A$.

\begin{definition}[Typical configurations]~
    \begin{enumerate}
    \item $\vec{\sigma}_0$: all spins in $A$ are $\SU$'s and all spins outside $A$ are $\SD$'s. $\mathcal{B}[\vec{\sigma}_0]=(\frac{D^3d^2-D}{D^4d^2-1})^{4l-2}(\frac{D^2d^2-D^2}{D^4d^2-1})\approx (\frac{1}{D})^{4l}$ is dominated by the boundary of $A$.
    \item $\vec{\sigma}_1$: all spins are $\SD$'s.
    $\mathcal{B}[\vec{\sigma}_1]=(\frac{D^4d-d}{D^4d^2-1})^{l^2}\approx (\frac{1}{d})^{l^2}$ is dominated by the bulk of $A$.
\end{enumerate}
\end{definition}
The two typical configurations reflect different effects.
A general configuration combines the boundary effect and the bulk effect.
Suppose $\vec{\sigma}$ has $t$ $\uparrow$ and $l^2-t$ $\downarrow$ inside $A$. One can prove that the perimeter is at least $4\sqrt{t}$ and $\mathcal{B}[\vec{\sigma}]\leq (1/d)^{l^2-t}(1/D)^{4\sqrt{t}}\leq \max\{(1/d)^{l^2}, (1/D)^{4l}\}$.
Therefore, the two typical configurations have the largest $\mathcal{B}$-value.
Moreover, any configuration with large $\mathcal{B}$ should resemble the two typical configurations, thus the number of such configurations are not too large. 
Consequently, $Z_p$ is upper bounded by $\mathcal{B}[\vec{\sigma}_0]+\mathcal{B}[\vec{\sigma}_1]$ with a modest factor, as shown in Theorem~\ref{theorem:entropy}.
On the other hand, by modifying the two typical configurations slightly, we explicitly construct configurations with large $\mathcal{B}$-value, resulting in the lower bound of Theorem~\ref{theorem:entropy}.

The key technique here is to bound the number of configurations by the generating function of polyominoes.
According to~\eqref{equ_Z_p}, all $\uparrow$ outside $A$ form several ESSs rooted at the left and top boundary of $A$, and all $\downarrow$ inside $A$ form several ESSs rooted at the right and bottom boundary of $A$.
In Appendix \ref{proof_theorem2}, we apply the bridge transformation to rigorously establish a relation between the number of configurations with the number of plane polyominoes.

When the bond dimension $D <D_c = d^{\frac{|A|}{|\partial A|}}$, the corresponding ground state configuration  is presented in Fig.~\ref{entropy_typicality} (a). 
We can bound the second R\'{e}nyi entanglement entropy:
\begin{equation}
    S_2(\rho_A) = |\partial A| \text{log} (D),
\end{equation}
which shows an area law \cite{Eisert2010Colloquium} for the tensor network states.  Similarly,  as shown in Fig.~\ref{entropy_typicality} (b), when $D>D_c$, the bulk part dominates, and the entanglement entropy will concentrate to the maximum.

\subsubsection{Maximum entropy for small subsystems}
For random MPS, it has been shown that the reduced density matrix of a small subsystem is concentrated to a maximum entropy with high probability \cite{collins2012matrix,Haferkamp2021Emergent,gonzalez2018spectral}.  For random high-dimensional tensor network, we show how the reduced density matrix $\rho_A$ approaches the maximum entanglement states with the increasing of the bond dimension $D$. 

\begin{theorem}
    \label{theorem:maximum_entropy} 
    Given a subset qubits system $A$ with volume $|A|=l\times l$ in the $V = L \times L$ tensor network with $L \rightarrow \infty$. Then for any physical dimension $d\ge 2$ and bond dimension $D\ge 3$, the averaged purity of the reduced density matrix $\rho_A$ can be upper bounded by the following inequality:
    \begin{equation}
        \text{Pr}\left[ \left|\text{Tr}(\rho_A^2)-d^{-|A|} \right| \geq \epsilon \right] \leq O\left(\frac{D^{-4}}{\epsilon}\right),
    \end{equation}
    where $\epsilon$ is a small constant.
    \end{theorem}
The proof of Theorem~\ref{theorem:maximum_entropy} follows a similar approach used to prove Theorem~\ref{theorem:entropy}. The only configuration with zero perimeter is $\vec{\sigma}_1$. The perimeters of other configurations are at least $4$. So $Z_p$ 
can be upper bounded by $d^{-|A|}+O(D^{-4})$. Theorem~\ref{theorem:maximum_entropy} then follows from the Markov inequality. For a disordered random MPS in one-dimensional cases, the concentration of purity on the maximal entangled states is bounded by $O(D^{-2})$ \cite{Haferkamp2021Emergent}, whereas for translation invariant MPS, it is bounded by $O(D^{-1/5})$ \cite{collins2012matrix,gonzalez2018spectral}. In contrast, the bound for 2D cases is $O(D^{-4})$, which is tighter than the one-dimensional case. Our technique applies straightforwardly for higher spatial dimensions as well. In 3D space, the surface area of any not-all-$\downarrow$ configuration is at least 6. So we expect that the concentration is bounded by $O(D^{-6})$. This result proves that the reduced density matrix is highly concentrated around the maximum entangled state, indicating that the subsystem is highly entangled with the rest of the system. We provide a detailed proof of this Theorem in Appendix~\ref{theorem3_proof}.

\subsection{Local typicality}
\label{local_typicality}
 Typicality refers to the phenomenon that a quantum many-body system's majority of pure microstates concentrate on a specific region of the Hilbert space. As a result, the measurement outcomes of these states are very similar to each other. In this work, we consider an observable $\hat{O}$, which is locally acted on the site $(x_0,y_0)$ of a random tensor network state. The expectation values for this local observable is $\langle \Psi | \hat{O}|\Psi \rangle$. For convenience, we assume that $\text{Tr}(\hat{O}) = 0$, so that $\overline{\langle  \Psi |\hat{O} |\Psi \rangle} = 0$. This does not diminish the generality as a generic observable can be written as the sum of a traceless observable and a constant factor multiple of the identity operator. We now prove that the expectation value of the local observable $\hat{O}$ is concentrated to $0$ with the bond dimension $D$. The variance of $\langle \Psi | \hat{O}|\Psi \rangle$ can be written as:
$ \text{Var}(\langle \Psi | \hat{O}|\Psi \rangle) = \overline{\langle \Psi | \hat{O}|\Psi \rangle^2}$. 
 We provide a bound on the variance of the expectation value in the following theorem:
 \begin{theorem}
    \label{theorem:typicality}
Given a $L \times L$ random tensor network state $|\Psi\rangle$ and an arbitrary local observable $\hat{O}$ acting on a single site, the expectation value of such a local observable obeys the following inequality:
\begin{equation}
    \text{Pr}\left[
        |\langle \Psi | \hat{O} | \Psi \rangle | \geq \epsilon
    \right] \leq O\left(\frac{\text{Tr}(\hat{O}^2) D^{-4}}{\epsilon^2}\right),
\end{equation}
where $\epsilon$ is an arbitrary constant real number. 
\end{theorem}
 
 As discussed above, we can map the calculation of variance to the calculation of the partition function for a classical Ising model illustrated in Fig. \ref{illustration_Methods} (d). The partition function reads:
 \begin{equation}
 \label{patition_function_typicality}
     Z_t = \sum_{\vec{\sigma}} t^{(2)}_{x_0,y_0} \prod_{\substack{ (x,y) \neq (x_0, y_0)}}  t^{(1)}_{x,y} ,
 \end{equation}
where $t^{(1)}_{x,y} $ is defined as:
 \begin{equation}
    \begin{aligned}
           t^{(1)}_{x,y} &= t^{(1)}_{x,y}(\sigma_{x,y}, \sigma_{x, y+1}, \sigma_{x+1,y}) \\
           &\coloneqq f(\sigma_{x,y}, \sigma_{x, y+1}, \sigma_{x+1,y}).\\
\end{aligned}
\end{equation}
Note that the definition of $f(\sigma_{x,y}, \sigma_{x, y+1}, \sigma_{x+1,y})$ is given in Eq.~(\ref{single_site_contraction}). The sites $(x,y)\neq(x_0,y_0)$ form a spin model with the same Hamiltonian as Eq.~(\ref{norm_ising_hamiltonian}). For $t^{(2)}_{x,y}$, we have:
\begin{equation}
\begin{aligned}
        t^{(2)}_{x,y} &= t^{(2)}_{x,y}(\sigma_{x,y}, \sigma_{x, y+1}, \sigma_{x+1,y}) \\&\coloneqq\ipic{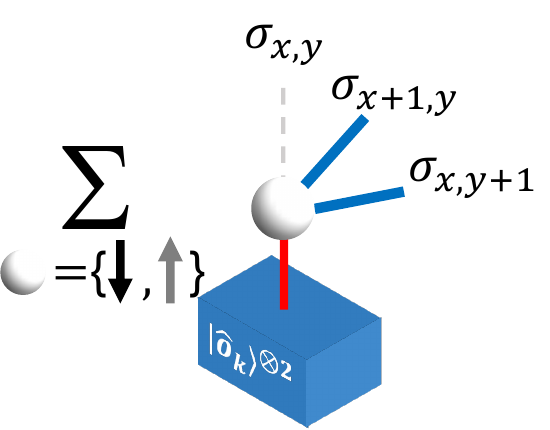}{0.4},
\end{aligned}
\end{equation}
where we denote $|\hat{O}_k\rangle^{\otimes2} = \hat{O}_k\otimes \hat{O}_k$ in a isomorphism form, and each $t^{(2)}_{x,y}(\sigma_{x,y}, \sigma_{x, y+1}, \sigma_{x+1,y})$ for different spin configurations $\sigma_{x,y} = \uparrow/\downarrow$ takes the value:
\begin{subequations}
    \begin{align}
        t^{(2)}_{x,y}(\downarrow,\downarrow,\downarrow) &= -\frac{\text{Tr}(\hat{O}^2)}{\left[(D^2d)^2-1\right]d},\label{ob1}\\
    t^{(2)}_{x,y}(\downarrow,\downarrow,\uparrow) &= t^{(2)}_{x,y}(\downarrow,\uparrow,\downarrow) = -\frac{D\text{Tr}(\hat{O}^2)}{\left[(D^2d)^2-1\right]d}, \label{ob2}\\
    t^{(2)}_{x,y}(\downarrow,\uparrow,\uparrow) &= -\frac{D^2\text{Tr}(\hat{O}^2)}{\left[(D^2d)^2-1\right]d},\label{ob3} \\
    t^{(2)}_{x,y}(\uparrow,\downarrow,\downarrow) &= \frac{D^2\text{Tr}(\hat{O}^2)}{(D^2d)^2-1},\label{ob4}\\
    t^{(2)}_{x,y}(\uparrow,\downarrow,\uparrow) &=t^{(2)}_{x,y}(\uparrow,\uparrow,\downarrow) = \frac{D^3\text{Tr}(\hat{O}^2)}{(D^2d)^2-1},\label{ob5}\\
    t^{(2)}_{x,y}(\uparrow,\uparrow,\uparrow) &= \frac{D^4\text{Tr}(\hat{O}^2)}{(D^2d)^2-1}.\label{ob6}
    \end{align}
    \end{subequations}

The ground state configuration is illustrated in Figure \ref{entropy_typicality} (c), where all spin states $\sigma_{x,y}$, except for $(x_0, y_0)$, are in the $\downarrow$ state. We can identify three types of excited state spin configurations: (1) The first type comprises a configuration that includes an ESS with a root at site $(x_0,y_0)$, as shown in Figure \ref{entropy_typicality} (d). In this case, the ESS's area should be $m\geq 0$, and we can utilize a similar bridge transformation described in Section~\ref{Toric_2_plane} to map this graph to a plane polyomino; (2) The second type of configuration involves a spin at either $\left((x_0+1)\%L, y_0\right)$ or $\left(x_0, (y_0+1)\%L\right)$ being in the $\uparrow$ state. To obtain a non-zero value, the spin configuration should include an ESS whose area is $m\geq L$, as stated in Properties~\ref{Properties_configurations}. The contribution of the site's local observable is determined by Eqs.(\ref{ob2},\ref{ob3},\ref{ob5},\ref{ob6}). The proof for this case can be directly derived from Theorem~\ref{theorem:normalized}, where the upper bound is given by an exponentially small factor $\kappa^{L}$ with respect to the system size $L$; (3) The third type is when both endpoints of the ESS are not at $(x_0, y_0)$. Similar to the second type, the contribution for this part is again exponentially small with respect to the system size. In subsequent calculations, we assume a sufficiently large system and thus neglect the effect from the second and third types ESSs.


We can write the parition function in Eq.~(\ref{patition_function_typicality}) in the following form:
\begin{equation}
    \begin{aligned}
        Z_t =& \sum_{\vec{\sigma}}  t^{(2)}_{x_0,y_0} \prod_{\substack{ (x,y) \neq (x_0, y_0)}} t^{(1)}_{x,y} ,\\
           = & \sum_{\vec{\sigma}}  t^{(2)}_{x_0,y_0}(\uparrow, \downarrow,\downarrow) \prod_{\substack{ (x,y) \neq (x_0, y_0)}}  t^{(1)}_{x,y}\\
          & + \sum_{\vec{\sigma}}  t^{(2)}_{x_0,y_0}(\downarrow, \downarrow,\downarrow)\prod_{\substack{ (x,y) \neq (x_0, y_0)}}  t^{(1)}_{x,y}. \\
    \end{aligned}
\end{equation}
 We find that the term $t^{(2)}_{x,y}(\downarrow, \downarrow,\downarrow)$ has a negative contribution, so that we can safely discard this term when calculating the upper bound. The partition function is upper bounded by:
\begin{equation}
\begin{aligned}
    Z_t \leq & t_{x_0-1,y_0}^{(1)}(\downarrow,\uparrow, \downarrow)t_{x_0,y_0-1}^{(1)}(\downarrow, \downarrow,\uparrow)t^{(2)}_{x_0,y_0}(\uparrow, \downarrow,\downarrow)\\
    &+G(\frac{1}{d}, \frac{1}{D^2})t^{(2)}_{x_0,y_0}(\uparrow, \downarrow,\downarrow), \\
\end{aligned}
\end{equation}
where the first term corresponds to the contribution from the ground state configuration, and the second term corresponds to the configurations contain ESSs with an endpoint at $(x_0, y_0)$. We can bound the effects of ESSs by the generating function in Eq.~(\ref{equ_D}). By taking into account the value of $t^{(2)}_{x,y}(\uparrow, \downarrow,\downarrow)$ and $t^{(1)}_{x,y}(\sigma_{x,y}, \sigma_{x, y+1}, \sigma_{x+1,y})$, the upper bound for the variance of $\langle \Psi | \hat{O} | \Psi \rangle$ is as follows:
\begin{equation}
   \text{Var}(\langle \Psi | \hat{O}|\Psi \rangle) \leq  O\left[\frac{1}{D^4}\text{Tr}(\hat{O}^2)\right].
\end{equation}
Using the Chebyshev's inequality, we arrive at Theorem \ref{theorem:typicality}. We can provide a lower bound by calculating the contribution of the ground state configuration, which also scales as $O\left[\text{Tr}(\hat{O}^2)/D^4\right]$. Therefore, this upper bound is tight.

From Theorem \ref{theorem:typicality}, when the bond dimension $D$ is large enough, the expectation value for arbitrary local observable $\hat{O}$ concentrates on a specific value. In the $1$D case, the bound for the variance of local observable expectation scales $D^{-2}$ with the bond dimension $D$~\cite{Haferkamp2021Emergent}. In contrast, the bound can be tighten to $D^{-4}$ in $2$D case. This means that the expectation value of a local observable will concentrate to a specific value more quickly as the bond dimension increases with respect to one dimension. One may also notice that this phenomenon has a profound connection to the barren plateau  problem in training a variational ansatz. As the measurement results for different variational states are almost the same, varying the variational parameters would not change the output result significantly. We will discuss this part in depth in Sec.~\ref{Sec:local_loss_function}.

\section{Barren plateaus in high-dimensional tensor network models}

\subsection{Global loss function}

First, we consider variational models based on tensor-network states with the global loss function:
\begin{equation}
\label{global_loss}
    \mathcal{L}_g = 1-\|\langle\Psi(\Theta)|\phi\rangle\|^2,
\end{equation}
where the state $|\phi\rangle$ is a normalized $n$-qubit constant quantum state. We remark that $|\phi\rangle$ can be an arbitrary state regardless of its entanglement properties. 
We define the tensor-network state $|\Psi(\Theta)\rangle$ as unitarily embedded random tensor-network states with variational parameters $\Theta$ (as shown in Fig.~\ref{illustration_peps} (b)). 
\begin{theorem}
\label{theorem:global} 
Denoting the derivative of the global loss function with respect to the variational parameter $\theta_k^{(i)}$ by $\partial_{k}^{(i)}\mathcal{L}_g$,  then  $\forall \;\theta_k^{(i)} \in \Theta$, $\partial_k^{(i)}\mathcal{L}_g$ obeys the following inequality:
\begin{equation}\label{Prob:global}
{\rm Pr}\left(|\partial_k^{(i)}\mathcal{L}_g|>\epsilon\right)\leq O\left(\frac{d^{-V}}{\epsilon^2}\right),
\end{equation}
where $\epsilon$ is an arbitrary constant and $\rm Pr(\cdot)$ denotes the probability.

\end{theorem}
For the expectation of such a derivative, it is straightforward to obtain $\overline{\partial_k^{(i)} \langle \uPEPS (\Theta) | \phi \rangle \langle \phi |\uPEPS (\Theta) \rangle } = 0$. Here we focus on the variance of such a derivative 
    \begin{equation}
    \begin{aligned}
               \var(\partial_k^{(i)}\mathcal{L}_g) &=\langle(\partial_k^{(i)}\mathcal{L}_g)^2\rangle - \langle\partial_k^{(i)}\mathcal{L}_g\rangle^2\\
               &= \overline{\left[\partial_\theta \langle \uPEPS (\Theta) | \phi \rangle \langle \phi |\uPEPS (\Theta) \rangle \right]^2} .
    \end{aligned}
    \end{equation}
    If we integrate each local unitary independently, according to Eq.~(\ref{eq_twodesign}) and Fig.~\ref{illustration_Methods} (c), the variance has the following spin representation, 
    \begin{equation}
    \begin{aligned}
            \label{upper_bound_cauchy_schawrz}
        \var(\partial_k^{(i)}\mathcal{L}_g) &\leq& \ipic{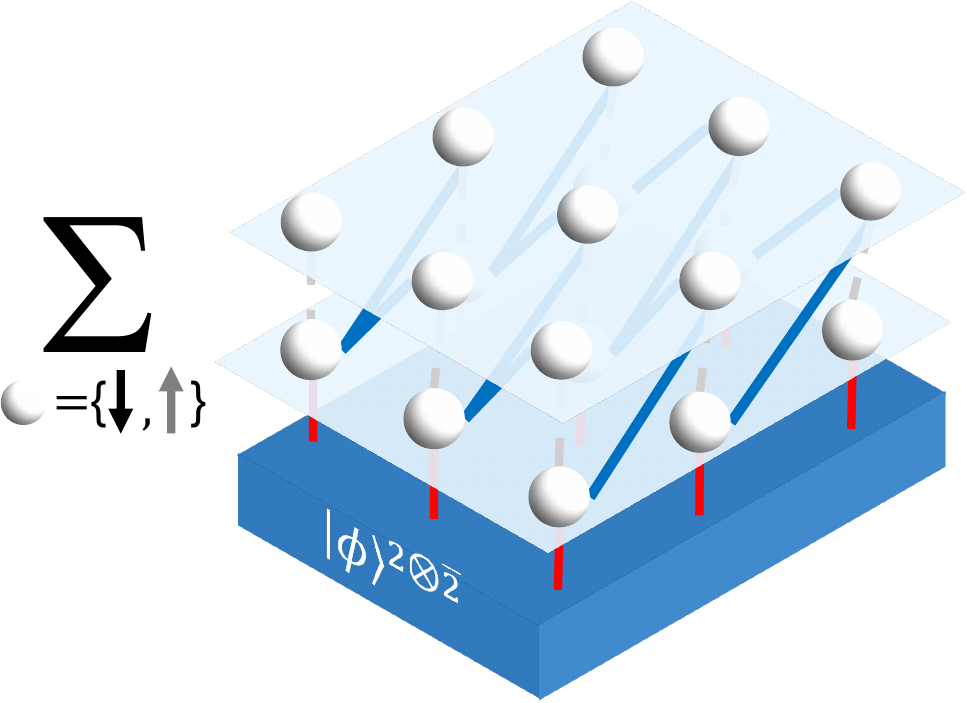}{0.3} C_k^{(i)},
    \end{aligned}
    \end{equation}   
    with a constant factor $C_k^{(i)}$. The constant factor comes from the contraction result of the site hosts derivative parameter, which is independent of the system size $L$, and we can upper bound it by a constant value. We further apply the Cauchy-Schwarz inequality\cite{kliesch2019Guaranteed}, so that the partition function of the spin model can be upper bounded by:
        \begin{equation}
    \begin{aligned}
            \label{upper_bound_cauchy_schawrz}
        \var(\partial_k^{(i)}\mathcal{L}_g) \leq \ipic{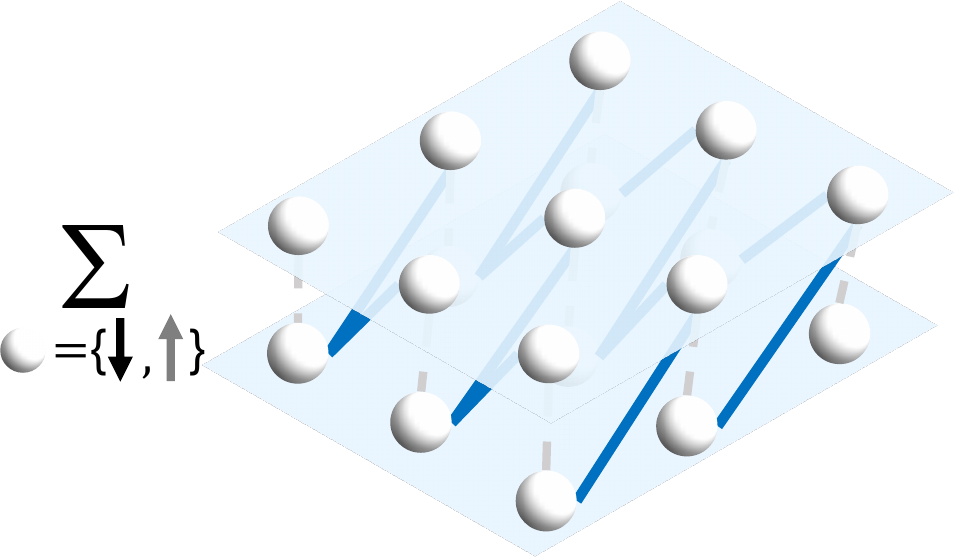}{0.3} C_k^{(i)}.
    \end{aligned}
    \end{equation} 
     The right term in Eq.~(\ref{upper_bound_cauchy_schawrz}) presents the partition function of a classical two-layer Ising model without an external field. The corresponding Hamiltonian is:
    \begin{equation}
    \begin{aligned}
        \globalH(\vec{\sigma}) =& \sum_{x,y} \globalH_{x,y}\left(\sigma_{x,y}^{(1)}, \sigma_{x,y}^{(2)}, \sigma_{x,y+1}^{(2)},\sigma_{x+1,y}^{(2)} \right)\\
       =&\sum_{x,y} \frac{1}{2} \left[\sigma_{x,y}^{(1)}(\sigma_{x,y+1}^{(2)} +\sigma_{x+1,y}^{(2)})J_1-\sigma_{x,y}^{(1)}\sigma_{x,y}^{(2)} J_2\right],
    \end{aligned}
    \end{equation}
    where the Ising coupling coefficient $J_2 = \text{i}\pi + \text{log}(D^2 d)$, $J_1 = - \text{log}(D)$. Our problem can be transferred to solve the partition function of the Ising model:
    \begin{equation}
        Z_g = \frac{1}{C}\sum_{\vec{\sigma}}e^{-\globalH(\vec{\sigma})},
    \end{equation}
    where $C =  [\frac{\text{i} D^2 d^{-\frac{1}{2}}}{(D^2 d)^2 -1 }]^{V}$. Although this kind of Ising model is exactly solvable, we can simply obtain an upper bound of its partition function without resorting to its exact solution. To this end, we  define:
    \begin{equation}
    \begin{aligned}
        g(\sigma_{x,y}^{(2)}, \sigma_{x,y+1}^{(2)},\sigma_{x+1,y}^{(2)})\coloneqq \ipic{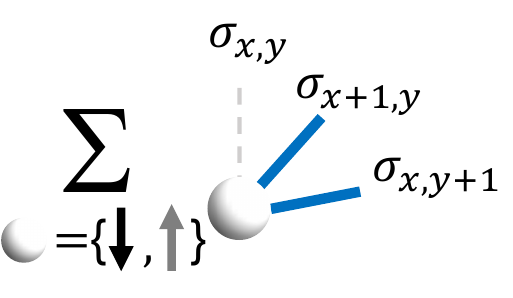}{0.4}, 
    \end{aligned}
    \end{equation}
    where the corresponding partition function can be written as:
    \begin{equation}
        Z_g = \frac{1}{C}\sum_{\vec{\sigma}}e^{-\globalH(\vec{\sigma})} = \sum_{\vec{\sigma}}\prod_{x, y}g(\vec{\sigma}_{x,y}^{(2)}, \vec{\sigma}_{x, y+1}^{(2)}, \vec{\sigma}_{x+1, y}^{(2)}).
    \end{equation}
    Direct calculations lead to:
    \begin{subequations}
         \begin{align}
        g(\downarrow,\downarrow,\downarrow) = g(\uparrow,\uparrow,\uparrow) =\frac{D^4 - 1/d}{(D^2d)^2 - 1},\\
        g(\downarrow,\downarrow,\uparrow) = g(\uparrow,\downarrow,\uparrow) = g(\downarrow,\uparrow,\downarrow) =  \frac{D^3 - D/d}{(D^2d)^2 - 1},\\
        g(\uparrow,\downarrow,\downarrow) = g(\downarrow,\uparrow,\uparrow) =\frac{D^2 - D^2/d}{(D^2d)^2 - 1},
    \end{align}
    \end{subequations}
    where the maximum value for a single site is $(D^4 - 1/d)/\left[(D^2d)^2 - 1\right]$. Hence, once we set the spin configuration of the upper layer, the partition function is smaller than $\left\{(D^4 - 1/d)/\left[(D^2d)^2-1\right]\right\}^{L^2-1}$. The number of the spin configurations of the upper layer is $2^{L^2}$. As a result, the summation in the right side of Eq. (\ref{upper_bound_cauchy_schawrz}) is bounded by:
    \begin{equation}
        \var(\partial_k^{(i)}\mathcal{L}_g) \leq C_k^{(i)} 2^{L^2}\left(\frac{D^4 - 1/d}{(D^2d)^2-1}\right)^{L^2-1},
    \end{equation}
    where we have $2(D^4-1/d)<D^4d^2-1$. At last, by using the Chebyshev's inequality, we conclude the Eq.~(\ref{Prob:global}).

The aforementioned Theorem \ref{theorem:global} has two interpretations. First, similar to the results in the one-dimensional case~\cite{liu2022presence}, the global loss function for the high-dimensional tensor networks also suffers from the barren plateau problem, namly the probability that the gradient along any direction is non-zero to any precision $\epsilon$  vanishes exponentially with respect to the number of physical sites. Second, the results demonstrate a connection between the trainability of tensor-network states and their global typicality. Global typicality refers to the phenomenon in which a quantum system, when prepared in a random state, has properties that are typical of most states chosen at random~\cite{Silvano2010Typicality}. Specifically, as the system size increases and becomes equilibrated, the expected values of observables in a random state converge to their typical values with a high degree of probability, which demonstrates the same characteristics of the barren plateaus phenomenon. 

\subsection{Local loss function}
\label{Sec:local_loss_function}

\begin{figure}[htp]
    \centering
    \includegraphics[scale=0.3]{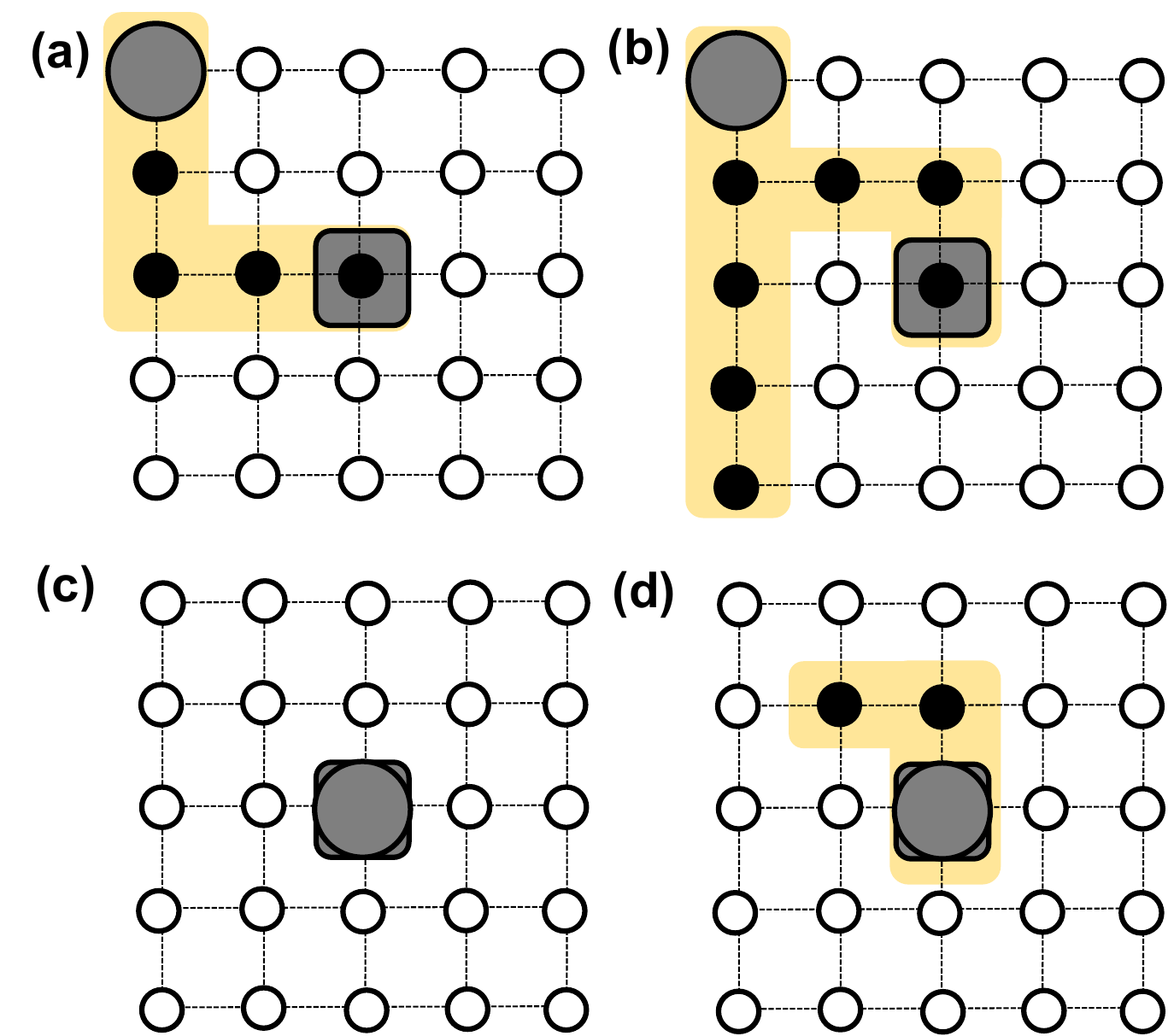}
    \caption{{Spin configurations for local loss functions.}  (a) the ground state spin configurations in the off-site case, where the grey discs represent the site hosts the derivative parameter, and the grey squares denote the site hosts the local observable. The ESS can only exist where it starts from the site hosts the derivative and ends at the site hosts the local observable. (b) the excited states spin configurations in the off-site case, where there exists at least one ESS whose length is larger than the $L$. (c) the ground state configurations in the on-site case, where the local observable acts on the same site that hosts the derivative parameter.  (d) the excited states spin configurations in the on-site case. The site that hosts the local observable can be the root of an ESS.}
    \label{local_spin_configurations}
    \end{figure}

Now, we consider the local loss function in the following form:
\begin{equation}
    \mathcal{L}_l = \langle \Psi(\Theta) | \hat{O}_i| \Psi(\Theta)\rangle,
\end{equation}
where $\hat{O}_i$ represents a local observable acts on the site $(x,y)$. Without loss of generality, we treat $\text{Tr}(\hat{O}^2)$, the physical dimension $d$, and the bond dimension $D$ as constant values and independent of the system size. We consider two cases for this problem. One is that two distinct sites host the derivative parameter and the local observable, i.e., the off-site case. The second one is that the same site hosts the derivative parameter and the local observable, i.e., the on-site case.

\subsubsection{The off-site case}

\begin{theorem}\label{theorem:local_loss} 
Denote the derivative of the local loss function $\mathcal{L}_l$ with respect to the parameter $\theta_k^{(i)}$ by $\partial_k^{(i)} \mathcal{L}_l$. Then $\forall\; \theta_k^{(i)}\in \Theta$, the variance of $\partial_k^{(i)} \mathcal{L}_l$  can be upper bounded by an exponentially decaying function with respect to the distance $\Delta_{\text{min}}$ between the site hosting observable and the site hosting the derivative parameter:
\begin{equation}
    \label{local_loss}
    \text{Pr}\left[|\partial_{k}^{(i)} \mathcal{L}_l| > \epsilon \right] \leq \mathcal{O} 
    \left(\frac{\kappa_l^{\Delta_{\text{min}}}}{\epsilon^2}\right),
    \end{equation}
    where the decay factor $\kappa_l = 0.93$.
\end{theorem}

It is straightforward to obtain that the expectation of such a derivative vanishes, similar to the one-dimensional case \cite{liu2022presence}. Therefore, we only need to focus on the variance of such a derivative:
 \begin{equation}
    \begin{aligned}
               \var(\partial_k^{(i)}\mathcal{L}_l) 
               &= \overline{(\partial_\theta \langle \Psi(\Theta) | \hat{O}_i| \Psi(\Theta)\rangle )^2}.
    \end{aligned}
    \end{equation}
{Similar techniques can be used to map the calculation of the variance to the partition function calculation of the Ising model on a lattice of size $L\times L$. With the exception of the site that hosts the derivative parameter, the spin model is the same as that in Section~\ref{local_typicality}. However, when considering the site that hosts the derivative parameter, we observe that the spin configuration must satisfy a specific restriction to obtain a non-zero configuration:} 
\begin{lemma}\label{lemma:supplementary_derivative_near} 
Consider a  $L\times L$ spin grid, where the site that hosts derivative parameter locates at $(x,y)$. Then, one of the spins at $(x,(y+1)\%L)$ or $((x+1)\%L, y)$ should be $\SU$ to obtain a non-zero configuration.
\end{lemma}
\begin{proof}
The proof of this lemma explores the same idea as in Ref. \cite{liu2022presence}. We define the unitary operators $U(\bm{\theta})$ generated by a sequential parameterized unitary gates: 
\begin{equation}
    U(\bm{\theta})=\prod_{\xi=1}^{{\rm Poly}(D^2 d)}e^{i\theta_{\xi}G_\xi}=U_{-}U_{+}, 
\end{equation}
where $G$ is a hermitian generator, and the corresponding parameter set is $\bm{\theta}=\{\theta_\xi\}$. We note that the parameters are chosen from a set of real random numbers. Later, we ignore the index $\xi$. The derivate of $\left[U(\bm{\theta})\otimes U^\dagger(\bm{\theta})\right]^{\otimes 2}$ with respect to the parameter $\theta$ is:
\begin{equation}
        \partial_{\theta} \left[
        \begin{array}{cc}
             & U(\bm{\theta}) \\
             & \otimes \\
             & U^\dagger(\bm{\theta})
        \end{array} \right]^{\otimes 2}=\sum_{\alpha,\beta = 0,1} (-1)^{\alpha + \beta + 1}\left(\begin{array}{cc}
             & U_-G^{\alpha}U_+ \\
             & \otimes \\
             & U_-^\dagger G^{1-\alpha}U_+^\dagger\\
             &\otimes\\
             &U_-G^{\beta}U_+\\
             & \otimes \\
             &U_-^\dagger G^{1-\beta}U_+^\dagger\\
        \end{array}
        \right),
\end{equation}
where the $\otimes$ acts on its upper and lower sides term. For a 2D grid, if we consider the spin value of $\sigma_{x,y+1} = \downarrow$ and $\sigma_{x+1,y}=\downarrow$, the integration in the site hosts the derivative parameter has the following relations:
\begin{equation}
    \label{derivative_term}
\begin{aligned}    
            &\sum_{\alpha,\beta=0,1} (-1)^{\alpha+\beta+1}\int dU_-dU_+ 
            \left[
            \begin{array}{cc}
                 &  (U_-G^{\alpha}U_+U_+^\dagger G^{1-\alpha}U_-^\dagger)\\
                 & \otimes \\
                 & (U_-G^{\beta}U_+U_+^\dagger G^{1-\beta}U_-^\dagger)\\
            \end{array} 
            \right]\\
            &=\sum_{\alpha,\beta=0,1} (-1)^{\alpha+\beta+1}\int dU_-dU_+ \left(
            \begin{array}{cc}
                 & U_-GU_-^\dagger \\
                 & \otimes \\
                 &U_-G U_-^\dagger
            \end{array}
            \right)\\
    &=0.
\end{aligned}
\end{equation}
 Hence, to obtain a non-zero configuration, one of the spins in $\left(x,(y+1)\%L\right)$ or $\left((x+1)\%L, y\right)$ should be $\SU$.
\end{proof}
Given the above lemma, we treat the partition function as a summation over two different kind of typical configurations $Z_l = Z_l^{(1)} + Z_l^{(2)}$. As shown in Fig.~\ref{local_spin_configurations} (a), the first configuration is that the ESS begins with the site hosting the observable and ends at the site hosting the derivative parameter. We can map this configuration to a plane polyomino directly (see the step 4 of the bridge transformation in Appendix~\ref{Toric_2_plane}), and the ESS has two definite endpoints. The minimal length of such an ESS is $\Delta_{\text{min}}$. Hence, we can bound this part of the partition function as:
\begin{equation}
\begin{aligned}
    Z_l^{(1)} &\leq O\left[\frac{1}{D^2} \sum_{m>\Delta_{\text{min}},n} D_{m,n}\left(\frac{1}{d}\right)^m\left(\frac{1}{D^2}\right)^n\right]\\
    & \leq O\left[\frac{1}{D^2} \left(\frac{1.86}{d}\right)^{\Delta_{\text{min}}}\sum_{m>\Delta_{\text{min}},n}D_{m,n}\left(0.54\right)^{m}\left(\frac{1}{D^2}\right)^n\right]\\
    & \leq O\left[\frac{1}{D^2} 0.93^{\Delta_{\text{min}}}G\left(0.54,\frac{1}{D^2}\right)\right] \leq O\left(0.93^{\Delta_{\text{min}}}\right).
\end{aligned}
\end{equation}

The second type of configuration in Fig.~\ref{local_spin_configurations} (b) corresponds to an ESS in the graph that does not end up with the site hosts local observable. In this case, each ESS at least has a length larger than $L$. This part of the partition functions $Z_l^{(2)}$ has the same bound with the Theorem \ref{theorem:normalized}. Indeed, we can ignore the effect of the term $Z_l^{(2)}$ when $L \gg \Delta_{\text{min}}$, and then, by using the Chebyshev's inequality, we prove the theorem.  As for a finite virtual bond dimension $D$, the derivative is upper bounded by an exponentially small number with respect to the minimal distance between the site hosting observable and the site hosting derivative parameter.

\begin{figure*}
    \centering
    \includegraphics[scale=0.2]{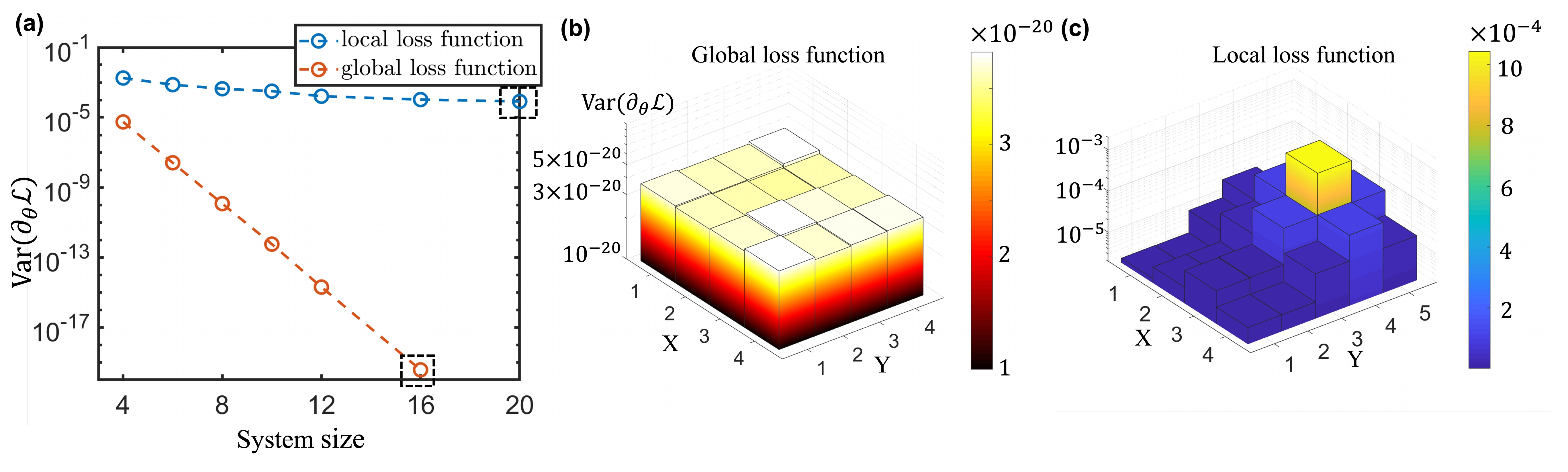}
    \caption{{Numerical results for global (local) loss function of tensor-network state.} (a) plots the average variance of $\partial_k\mathcal{L}_g$ and  $\partial_k\mathcal{L}_l$ with different sizes of the tensor networks, including $\{(2\times 2),(2 \times 3),(2\times 4),(2\times 5),(3\times 4),(4\times 4),(4\times 5)\}$ (note that we only plot up to $(4\times 4)$ for the global loss function). For the global loss function, we plot the mean value of the variance over all the variational parameters. We also plot the maximal variance of all the variational parameters for the local loss function.  (b) the distribution of the variance in different sites for a global loss function for the tensor-network state with a size of $(4\times 4)$. (c) The distribution of the average variance in different sites for a local loss function, where the site hosts local observable obtain the maximal variance for the tensor-network state with a size of $(4\times 5)$. Here we fix the physical dimension $d=2$ and the bond dimension $D=2$.}
    \label{figure_numerical_PEPS}
    \end{figure*}

\subsubsection{The on-site case}
\begin{theorem}\label{theorem:local_loss_2} 
    Denote the derivative of the local loss function $\mathcal{L}_l$ with respect to the parameter $\theta_k^{(i)}$ by $\partial_k^{(i)} \mathcal{L}_l$. We assume that the $\text{Tr}(\hat{O}^2)$ and physical bond dimension $d$ are constant. Then $\forall\; \theta_k^{(i)}\in \Theta$, the variance of $\partial_k^{(i)} \mathcal{L}_l$  has a lower bound:
    \begin{equation}
        \label{local_loss_2}
        \text{Var}(\partial_{k}^{(i)} \mathcal{L}_l) \geq O(poly(\frac{1}{D})),
    \end{equation}
    where $\text{poly}(\frac{1}{D})$ is a polynomial of bond dimension $D$ with constant degree. 
    \end{theorem}

In this case, the ground state and excited state spin configurations are presented in Fig.~\ref{local_spin_configurations} (c, d). It can be easily verified that the value of each excited state configuration is non-negative. Hence, the partition function can be lower bounded by the value of the ground state configuration:
\begin{eqnarray}
     \sum_{\alpha,\beta=0,1} (-1)^{\alpha+\beta+1} \int dU_H 
     \left[
            \begin{array}{cc}
                 &  \langle 0|\text{Tr}_{D}(U_-G^{\alpha}(\hat{O}') ^\dagger G^{1-\alpha}U_-^\dagger)|0\rangle\\
                 & \otimes \\
                 &\langle 0|\text{Tr}_{D}(U_-G^{\beta} \hat{O}'G^{1-\beta}U_-^\dagger)|0\rangle\\
            \end{array} 
            \right],\nonumber
\end{eqnarray}
where $\text{Tr}_{D}(\cdot)$ traces out the bond dimension of each local tensor, and $\hat{O}' = U_+ \hat{O} U_+^\dagger$. If $U_+$ forms $2$-design, we have
\begin{equation}
\begin{aligned}
            \text{Var}(\partial_{k}^{(i)}\mathcal{L}_l) 
            \geq \frac{D^2\text{Tr}(\hat{O}^2)}{(D^2d)^2 -1}C_+,
\end{aligned}
\end{equation}
where we define 
\begin{equation*}
    C_+ = \sum_{\alpha,\beta=0,1}\int dU_H \left(\rho_- G^{1-\alpha+\beta}\rho_- G^{1-\beta+\alpha} \right),
\end{equation*}
with $\rho_- = U_-(|0\rangle \langle 0 | \otimes \mathbb{I}_{D^2})U_-^\dagger$, and $\text{Dim}(\mathbb{I}_{D^2}) = D^2\times D^2$. Similarly, if $U_-$ forms $2$-design, we have
\begin{equation}
\begin{aligned}
            \text{Var}(\partial_{k}^{(i)}\mathcal{L}_l) \geq \frac{D^2 (1-\frac{1}{d})}{(D^2d)^2 -1} C_-,
\end{aligned}
\end{equation}
and we define the factor
\begin{equation*}
    C_- = 
    \sum_{\alpha,\beta=0,1}\int dU_H \text{Tr}\left(\hat{O}'G^{1-\alpha+\beta}\hat{O}'G^{1-\beta+\alpha}\right),
\end{equation*}
where $\hat{O}' = U_+ \hat{O} U_+^\dagger$. Hence, the value of such spin configurations are independent of the system size. The partition function of the on-site case is lower bound by a constant number. Hence, the gradients do not vanish and there is no barren plateaus problem in this case.

\section{Numerical results}
\label{Numerical_Results}

Our proof mainly focuses on the scaling behavior of the properties. As a complement to our analysis, we also perform numerical benchmarks on systems of modest size. In our numerical simulations, we utilize open-source libraries , TensorFlow \cite{Abadi2016TensorFlow} and TensorNetwork \cite{roberts2019tensornetwork}.
The simulations involved a general tensor-network state with periodic boundary conditions, and the global and local loss functions were compared in Fig.~\ref{figure_numerical_PEPS} (a). The global loss function is defined as $\mathcal{L}_{g} = 1-|\langle \phi | \Psi(\Theta)\rangle|^2/Z$, where $Z = \langle \Psi(\Theta) | \Psi(\Theta) \rangle$ is a normalization factor, and $|\phi\rangle = |\phi_k\rangle^{\otimes L^2} = \left[\frac{1}{\sqrt{2}}(|0\rangle+|1\rangle)\right]^{\otimes L^2}$ is a tensor product state. On the other hand, the local loss function was defined as $\mathcal{L}_{l} = \langle \Psi(\Theta)| \mathbb{I}^{\otimes_{i=1}^{k-1}} |\phi_k\rangle\langle\phi_k|\mathbb{I}^{\otimes_{i=k+1}^{L^2}} |\Psi(\Theta)\rangle/Z$ with the same normalization factor as the global loss function, where $k$ is the site index for applying the local operator.

Based on our numerical results, we observe that as the system size increases, the global loss function exhibits an exponential decay in $\text{Var}(\partial_\theta \mathcal{L}_g)$. On the other hand, for the local loss function, $\text{Var}(\partial_\theta \mathcal{L}_l)$ converges to a finite value. Those results are consistent with Theorem~\ref{theorem:global} and Theorem~\ref{theorem:local_loss_2}. In Figure~\ref{figure_numerical_PEPS} (b) and (c), we present the distribution of $\text{Var}(\partial_\theta \mathcal{L})$ across various sites. For the global loss function, we have found that $\text{Var}(\partial_\theta \mathcal{L}_g)$ is uniformly distributed across all sites. However, for the local loss function, the site corresponding to the local observable exhibits the highest variance, where the variance reduces as the distance between the site hosting the local observable and the site hosting the derivative parameter increases. We note that the decay behavior may not strictly follow an exponential trend. Since in the 2D case, there are multiple paths between the sites that host derivative parameters and those that host local observable. Additionally, paths with lengths longer than the minimum distance between these two sites also have an impact on the results in a finite-size system. Nonetheless, our numerical results confirm the analytical results presented in Theorem \ref{theorem:local_loss}.

\section{Discussion}
\label{Discussions}
Many important questions remain unexplored and deserve further investigations.
First, it would be worthwhile to extend our random tensor network techniques to more specific models, such as the quantum generative models based on projected entangled pair states \cite{Gao2018Quantum}. This could provide further insights into the capabilities and limitations of these models for machine learning applications. Second, exploring the statistical properties of tensor-network states with different structures, including the multi-scale entanglement renormalization ansatz and tree tensor-network states, would be an interesting direction. Such studies could lead to a better understanding of the relationship between tensor-network states and their physical properties.
{Recent studies in the field of random quantum circuits~\cite{napp2022efficient, gao2021limitations, nahum2018operator, hunter2019unitary, von2018operator, zhou2019emergent} has revealed that similar challenges arise when computing the partition functions of the Ising model. Therefore, our research methods may provide more insights to enhance the current understanding of this field.}


In conclusion, we have carried out a relatively comprehensive investigation into the statistical properties of high-dimensional random tensor networks and the trainability of variational tensor networks. We develop a combinatorial method based on solving the "puzzle of polyominoes" to rigorously study statistical properties of high-dimensional random tensor networks. We have proved that for sufficiently large bond dimensions, the entanglement entropy can approach a near-maximal value, and the typicality occurs for the expectation value of a local observable. In addition, we investigate the barren plateaus for high-dimensional tensor network models and prove that such models suffer from barren plateaus for global loss functions, making their training processes inefficient. However, for local loss functions, the gradient is independent of the system size but decays exponentially with the distance between the region where the local observable acts and the site that hosts the derivative parameter. These results provide valuable insights into the behavior of variational high-dimensional tensor networks, which would provide a valuable guide for  future theoretical studies and practical applications.

\begin{acknowledgements}
We acknowledge helpful discussions with Ignacio Cirac, Xun Gao, Xiaoliang Qi, Zhengzhi Sun, Shunyao Zhang, Xiaoqi Sun, Zhiyuan Wei, and Zhiyuan Wang. 
This work was supported by the National Natural Science Foundation of China (Grants No. 12075128 and T2225008),  Shanghai Qi Zhi Institute,
the Innovation Program for Quantum Science and Technology (No. 2021ZD0301601), the
Tsinghua University Initiative Scientific Research Program, and the Ministry of Education of China.
\end{acknowledgements}

\appendix

\section{Details in the Proof of Theorem 1}
\label{proof_theorem1}
First we introduce some basic notations in graph theory.
A graph $G$ is two sets $(V, E)$.
Elements in $V$ are called vertices of $G$.
$E\subseteq V\times V$ is a set of ordered pairs of vertices.
A pair $(u, v)\in E$ is called a directed edge from $v_1$ to $v_2$.
We also call $(u, v)\in E$ an in-edge of $v$ and an out-edge of $u$.

We say two vertices $u, v\in V$ are connected by a undirected path if there is a sequence of vertices $u=v_0-v_1-\cdots-v_{t-1}-v_t=v$ such that for any $0\leq i\leq t-1$, there is an edge between $v_i$ and $v_{i+1}$ (either $(v_i, v_{i+1})$ or $(v_{i+1}, v_i)$).
This is an equivalence relation so $V$ is divided into several equivalence classes.
Each class is called a connected component of $G$. 
The definition of connected component here is called ``weak connectivity" in graph theory since we ignore the direction of edges.

We say a vertex $u\in V$ can reach $v\in V$ by a directed path if there is a sequence of vertices $u=v_0\to v_1\to\cdots\to v_{t-1}\to v_t=v$ such that $(v_i, v_{i+1})\in E$ for all $0\leq i\leq t-1$.

\subsection{From Toric Polyominoes to Plane Polyominoes}
\label{Toric_2_plane}
The definition of toric polyominoes is very similar to plane polyominoes. The difference is that a toric polyomino is a type of geometric shape that is defined on a torus, which is a doughnut-shaped object. Unlike plane polyominoes, a toric polyomino may have multiple connected components.

\begin{figure*}
    \centering
    \includegraphics[scale=0.4]{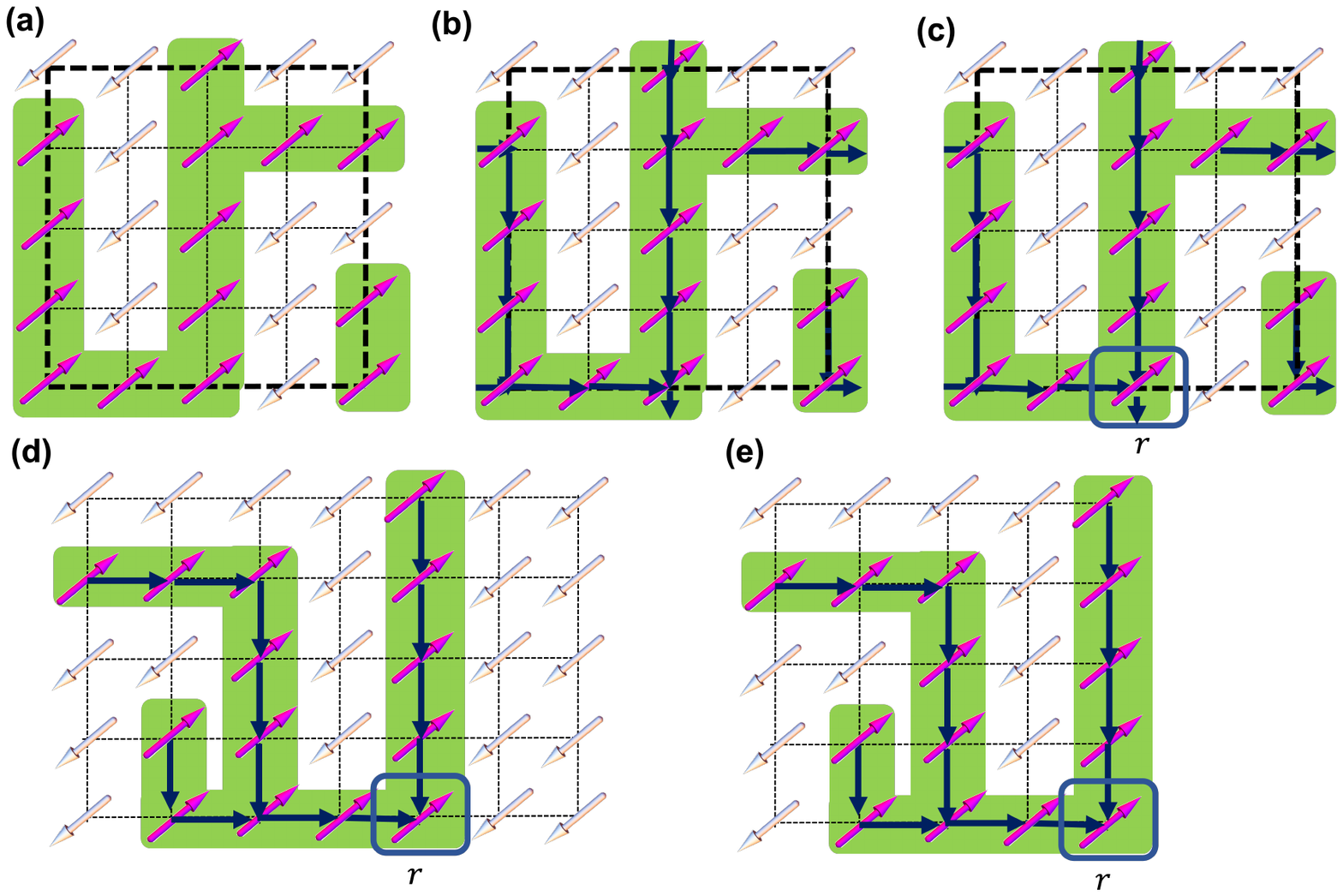}
    \caption{Illustration of the bridge transformation.
    (a) An illustration of a specific type of the toric polyomino with area $m = 14$ and upper perimeter $n = 5$. The bold dashed line represents the periodic boundary.
    (b) Construct a directed graph: For a spin-up site $(x, y)$, one of $((x+1)\%L, y)$ and $(x, (y+1)\%L)$ is spin-up. 
    If $((x+1)\%L, y)$ is spin-up, connect an edge from $(x, y)$ to $((x+1)\%L, y)$.
    Otherwise if $(x, (y+1)\% L)$ is spin-up, connect an edge from $(x, y)$ to $(x, (y+1)\% L)$.
    This example has only one connected component, so step 2 is automatically done.
    (c) Find a root: We set a root vertex $r$ such that every vertex in the connected component can reach $r$ through a directed path.  
    (d) Convert to a plane polyomino: We can map the toric polyomino by tracking back the path from the root $r$. 
    (e) Move the root to a fixed site. }
    \label{fig_bridge_transformation}
\end{figure*}

Here we introduce an algorithm, named bridge transformation, that is used to decompose connected components of a toric polyomino and transform them into plane polyominoes. Our algorithm has four steps:
\subsubsection*{Step 1: Construct a directed graph}
    Let $\vec{\sigma}$ be a toric polyomino.
    We construct a directed graph $G=(V, E)$ where the vertex set $V$ is the set of all $\uparrow$-sites.
    For an $\uparrow$-site $(x, y)$:
    \begin{itemize}
        \item If $\sigma_{x+1, y}=\uparrow$, we connect an directed edge from $(x, y)$ to $(x+1, y)$.
        \item If $\sigma_{x+1, y}=\downarrow, \sigma_{x, y+1}=\uparrow$, we connect a directed edge from $(x, y)$ to $(x, y+1)$.
    \end{itemize}    
    In this section $\vec{\sigma}$ is a always a toric polyomino, meaning that at least one of $\sigma_{x+1, y}$, $\sigma_{x, y+1}$ is $\uparrow$, so every vertex has exactly one out-edge.
    
    In the proof of Theorem~\ref{theorem:entropy} and Theorem~\ref{theorem:maximum_entropy}, a configuration $\vec{\sigma}$ may not be a toric polyomino, so it is possible that some vertex has no out-edge.

\subsubsection*{Step 2: Decompose the directed graph into connected components}
    The directed graph may have multiple connected components and we will deal with them separately.
\subsubsection*{Step 3: Find a root for every connected components}
    We claim that, in every connected component, we can find a vertex $r$ such that every vertex in this connected component can reach $r$ through a directed path. We call $r$ the root of this connected component.

    Indeed, if there is a vertex in the connected component with no out-edge, then it must be the root of this connected component.
    We will call this connected component an ended ESS.

    Now assume every vertex has a unique out-edge.
    Starting from an arbitrary vertex, one can move along the out-edge forever, so there must be a cycle.
    Now let $v\in V$ in the same connected component.
    Due to weak connectivity, there is an undirected path from $v$ to the cycle, denoted by $v-v_1-\cdots-v_n$. Here $v_1,\cdots, v_{n-1}$ are distinct vertices outside the cycle and $v_n$ is in the cycle.
    Since the unique out-edge of $v_n$ is in the cycle, the direction of $v_{n-1}-v_n$ must be $v_{n-1}\to v_{n}$.
    Similarly, the direction of $v_{n-2}-v_{n-1}$ must be $v_{n-1}\to v_{n-1}$. 
    We finally infer that the path from $v$ to $v_n$ is actually a directed path.
    In other words, any vertex in the connected component can reach the cycle via a directed path.
    Therefore, any vertex in the cycle can be the root.
    We will call this connected component a cycle ESS.

    In the proof of Theorem 1, $\vec{\sigma}$ is a toric polyomino so every ESS is a cycle ESS.

\subsubsection*{Step 4: Convert every ESS to a plane polyomino}
    In an ESS, every vertex has a directed path to the root.
    The path is indeed unique because every vertex has only one out-edge.
    We can recover a plane polyomino by backtracing the path from the root.
    Mathematically, we map the root to the site $(0,0)$ in a plane.
    For every vertex $v$ in the ESS, assume that $v$ can reach $r$ through $x$ right moves and $y$ down moves, then we map $v$ to the site $(-x, -y)$ in the plane.
    The resulting shape in the plane is a plane polyomino.

Fig.~\ref{fig_bridge_transformation} gives a concrete example for these steps.
We analysis the output of the bridge transformation in the following lemma.

\begin{lemma}[Output of the bridge transformation]\label{lem_output_of_bridge_transformation}
    Let $\vec{\sigma}$ be a toric polyomino with area $m$ and upper perimeter $n$.
    The output of the bridge transformation is $k$ plane polyominoes with area $m_i$ and upper perimeter $n_i~(i=1,2,\cdots, k)$, respectively.
    We have $k\leq m/L\leq L$ and 
    \begin{align}
        &m_1, \cdots, m_k\ge L\\
        &m_1+\cdots +m_k=m\\
        &n\leq n_1+\cdots +n_k\leq n+k \label{equ_perimeter_increment}
    \end{align}
\end{lemma}
\begin{proof}
    It is easy to see that the bridge transformation does not enlarge the area, so $\sum_{i=1}^k m_i=m$.
    In step 3 we prove that every ESS contains a cycle, and the size of the cycle is at least $L$, implying $m_i\ge L$.
    As a conclusion, $k\leq m/L\leq L$.

    Now we analysis the change of upper perimeter.
    Suppose in step 2, we decompose $\vec{\sigma}$ into $k$ ESSs with area $m_i$ and upper perimeter $n_i'$, respectively.
    And suppose in step 4, we convert the $i$-th ESS to a plane polyomino with area $m_i$ and upper perimeter $n_i$, respectively.

    The first observation is that $\sum_{i=1}^k n_i'=n$.
    This benefits from our construction of the directed graph in step 1: if $\sigma_{x,y}=\sigma_{x+1,y}=\uparrow$, the arrow on site $(x,y)$ is pointing to $(x+1,y)$, thus they are in the same ESS
    That is to say, we never separate vertical neighbors in step 2, so the upper perimeter will not increase.

    The other observation is $n_i'\leq n_i\leq n_i'+1$.
    Indeed, if a vertical neighbor $(x,y)$ and $(x+1,y)$ is separated in the step 4, $(x,y)$ must be the root of the ESS.
    If this increment occurs, the ESS must be a cycle ESS.
    Therefore, for every cycle ESS, the upper perimeter is increased by at most 1.
    Combining the two observation, we obtain $n\leq \sum_{i=1}^k n_i \leq n+k$.
\end{proof}

Now we are able to bound the number of toric polyominoes, $\widetilde{D}_{m,n}$, by the number of plane polyominoes, $D_{m,n}$.
\begin{lemma}\label{lem_toric_and_plane}
    \begin{equation}
        \widetilde{D}_{m,n}\leq \sum_{k=1}^{[m/L]}\sum_{c=0}^{k}\sum_{\substack{m_1,\cdots,m_k\ge L,\\m_1+\cdots+m_k=m,\\n_1+\cdots n_k=n+c}}\prod_{i=1}^{k}(L^2D_{m_i,n_i})
    \end{equation}
\end{lemma}
\begin{proof}
     This is a directed corollary of lemma~\ref{lem_output_of_bridge_transformation}. 
    For fixed $k$ and $(m_i,n_i)$'s, there are $\prod_{i=1}^k D_{m_i,n_i}$ possible shapes of the $k$ plane polyominoes output by the bridge transformation.
    Each plane polyomino has $L^2$ possible locations before step 4.
    So the number of possibilities is at most $\prod_{i=1}^k (L^2D_{m_i,n_i})$.
\end{proof}

\subsection{Reflection Trick}
    In \eqref{partition_upperbound}, we have seen that $Z\leq 1+\sum_{m,n}\widetilde{D}_{m,n}q_a^mq_u^{n}=1+\sum_{m,n}\widetilde{D}_{m,n}q_a^mq_p^{2n}$.
    This inequality can be simply improved via a reflection argument.
    The intuition is that, in lemma~\ref{lem_area_and_perimeter_bound}, we use the fact $p\ge 2n$ to deduce $q_a^mq_p^p\leq q_a^mq_p^{2n}$.
    However, this step only considers the effect of vertical perimeter and ignores horizontal perimeter.
    To reduce this waste, if the upper perimeter of $\vec{\sigma}$ is larger than the left perimeter, we make the reflection of $\vec{\sigma}$ about the diagonal, i.e., construct $\vec{\sigma}'$ such that $(\vec{\sigma}')_{x,y}=\vec{\sigma}_{y,x}$.
    With the reflection, we only need to concern configuration whose upper perimeter is no more than the left perimeter, so $p\ge 4n$ and $q_a^mq_p^p\leq q_a^mq_p^{4n}$.
    Therefore,
    \begin{equation}
        Z\leq 1+2\sum_{m,n}\widetilde{D}_{m,n}q_a^m q_p^{4n}.
    \end{equation}

\subsection{Proof of the Theorem}
\label{proof of norm}
By lemma~\ref{lem_toric_and_plane},
\begin{align}
    &\sum_{m,n}\widetilde{D}_{m,n}q_a^mq_p^{4n} \nonumber\\    \leq&\sum_{m,n}\sum_{k=1}^{[m/L]}\sum_{c=0}^{k}q_p^{-4c}\sum_{\substack{m_1,\cdots,m_k\ge L,\\m_1+\cdots+m_k=m,\\n_1+\cdots n_k=n+c}}\prod_{i=1}^{k}(L^2D_{m_i,n_i}q_a^{m_i}q_p^{4n_i})\nonumber\\
    \leq& \sum_{k=1}^{L}\sum_{c=0}^{k}q_p^{-4c}\sum_{m,n}\sum_{\substack{m_1,\cdots,m_k\ge L,\\m_1+\cdots+m_k=m,\\n_1+\cdots n_k=n+c}}\prod_{i=1}^{k}(L^2D_{m_i,n_i}q_a^{m_i}q_p^{4n_i})\nonumber\\
    \leq& \sum_{k=1}^{L}\frac{q_p^{-4k}}{1-q_p^4}\left(L^2\sum_{m\ge L,n}D_{m,n} q_a^mq_p^{4n}\right)^k\nonumber\\
    \leq& \frac{L}{1-q_p^4}\max_{k\leq L}\left(L^2q_p^{-4}\sum_{m\ge L,n}D_{m,n} q_a^mq_p^{4n}\right)^k.\label{equ_bound_on_Z_1}
\end{align}
By definition, $q_a\leq 1/2, q_p^4 \leq 1/16$. So
\begin{align}
    &q_p^{-4}\sum_{m\ge L, n}D_{m,n} q_a^mq_p^{4n}\nonumber\\
    \leq& \sum_{m\ge L, n}D_{m,n} (0.5)^{m}(1/16)^{n-1}\nonumber\\
    \leq& 16(0.5/0.72)^L\sum_{m\ge L, n}D_{m,n} (0.72)^m(1/16)^n\nonumber\\
    \leq& 16(0.695)^{L}G(0.72, 1/16)\leq 27(0.695)^L.\label{equ_truncation_decay_with_L}
\end{align}
For $L\ge 30$, $27L^2(0.695)^L<1$, so
\begin{equation}
    \sum_{m,n}\widetilde{D}_{m,n}q_a^mq_p^{4n}\leq \frac{27L^3}{1-q_p^4}(0.695)^L<c(0.7)^L
\end{equation}
for some large constant $c$.
We can make $c$ large enough such that the right hand side of \eqref{equ_bound_on_Z_1} is less than $c(0.7)^L$ for $L<30$ as well.
Hence 
$$1\leq Z<1+c_1(0.7)^L$$
for a universal constant $c_1$.
Theorem 1 is proved.

\section{Details in the Proof of Theorem~\ref{theorem:entropy}}
\label{proof_theorem2}
\subsection{Outline}
In this section we bound
\begin{equation}
    Z_p=\sum_{\vec{\sigma}}\mathcal{B}[\vec{\sigma}],
\end{equation}
where $\mathcal{B}[\vec{\sigma}]$ is
\begin{small}
\begin{equation*}
    \prod_{(x, y)\not\in A}f(\sigma_{x, y}, \sigma_{x, y+1}, \sigma_{x+1, y})\prod_{(x, y)\in A}f(-\sigma_{x, y}, -\sigma_{x, y+1}, -\sigma_{x+1, y}).
\end{equation*}
\end{small}

Recall the definition of $f$,
\begin{equation}\label{equ_entanglement_single_site}
    f(x, y, z)\leq q_a^{\mathcal{I}(x=\uparrow)}q_p^{\mathcal{I}(x\neq y)+\mathcal{I}(x\neq z)},
\end{equation}
where $\mathcal{I}(P)$ is the indicator function of an event $P$.
Multiplying~\eqref{equ_entanglement_single_site} for all $(\sigma_{x,y}, \sigma_{x, y+1}, \sigma_{x+1,y})$, we can bound $\mathcal{B}[\vec{\sigma}]$ by its twisted area and perimeter:
\begin{lemma}\label{lem_upper_bound_of_B}
    Given a configuration $\vec{\sigma}$.
    Define its twisted area $m$ as the number of $\uparrow$ outside $A$ plus the number of $\downarrow$ inside $A$.
    Define its perimeter $p$ and upper perimeter $n$ as in the main text.
    Then 
    \begin{equation}\label{equ_bound_for_B}
        \mathcal{B}[\vec{\sigma}]\leq q_a^mq_p^{p}\leq q_a^mq_p^{2n}.
    \end{equation}
\end{lemma}

The upper bound for $\mathcal{B}$ is similar to $\mathcal{A}$.
Indeed, for most configurations, the effect of region $A$ is relatively insignificant, so we can utilize theorem~\ref{theorem:normalized} to show that their contributions to $Z_p$ are negligible.

$Z_p$ is mainly contributed by those small configurations around region $A$ (by small we mean that the twisted area and the perimeter are small).
Two typical examples are $\vec{\sigma}_0$ and $\vec{\sigma}_1$ defined in the main text. $\mathcal{B}[\vec{\sigma}_0]\approx D^{-4l}$ and $\mathcal{B}[\vec{\sigma}_1]\approx d^{-l^2}$ exhibit the boundary effect and the bulk effect, respectively.
We will prove that every configuration with large $\mathcal{B}$ value must resemble $\vec{\sigma}_0$ or $\vec{\sigma}_1$ and the number of such configurations are not too large.
As a conclusion, $Z_p$ is upper bound by $\mathcal{B}[\vec{\sigma}_0]+\mathcal{B}[\vec{\sigma}_1]$ with a modest factor. 
This gives the upper bound in Theorem~\ref{theorem:entropy} and Theorem~\ref{theorem:maximum_entropy}.

On the other hand, we explicitly construct a family of configurations by modifying $\vec{\sigma}_0$ and $\vec{\sigma}_1$ slightly.
The $\mathcal{B}$ value of these configurations are similar to $\mathcal{B}[\vec{\sigma}_0]$ and $\mathcal{B}[\vec{\sigma}_1]$.
Adding up all these $\mathcal{B}$ value, we obtain the lower bound of $Z_p$ of order $\mathcal{B}[\vec{\sigma}_0]+\mathcal{B}[\vec{\sigma}_1]$.

We clarify some notations and assumptions here.
Let $\Sigma^A$ be the set of configurations with nonzero $\mathcal{B}$-values.
Let $E_{m,n}$ be the number of configurations in $\Sigma^A$ with area $m$ and upper perimeter $n$.
Let $E(x, y)=\sum_{m,n}E_{m,n}x^my^n, x, y\in (0, 1]$.
By Lemma~\ref{lem_upper_bound_of_B} and the reflection trick, $$Z_p\leq 2E(q_a, q_p^4)=2\sum_{m,n}E_{m,n}q_a^mq_p^{4n}.$$
We say a configuration is large (small, resp.) if its twisted area is at least (less than, resp.) $L/2$.
Denote $Z_{p, \text{large}}=\sum_{\vec{\sigma}\in \Sigma^A, \text{large}}\mathcal{B}[\vec{\sigma}]$, $Z_{p, \text{small}}=\sum_{\vec{\sigma}\in \Sigma^A, \text{small}}\mathcal{B}[\vec{\sigma}]$.
Let $\Sigma^A_{\text{small}}$ be the set of small configurations in $\Sigma^A$.

Let $\partial^{l}A=\{(x, y)\not\in A:(x, y+1)\in A\}$ be the left boundary of region $A$, $\partial^{t}A=\{(x, y)\not\in A:(x+1, y)\in A\}$ be the top boundary, $\partial^{b}A=\{(x, y)\in A:(x, y+1)\not\in A\}$ be the bottom boundary, and $\partial^{r}A=\{(x, y)\in A: (x, y+1)\not\in A, (x+1, y)\in A\}$ be the right boundary.

We assume $d, D\ge 2$ and $L$ is sufficiently large (larger than a fixed function $h(d, D, l)$, $h$ will be specified in many places in the proof).

\subsection{Configurations with Large Twisted Area}
We first settle down to large configurations, where a configuration is large if and only if its twisted area is at least $L/2$.
As the twisted area is large compared to region $A$, the effect of $A$ is insignificant, so the proof is similar to the Theorem 1.

We further eliminate the details in $A$ by change every site in $A$ to $\uparrow$.
Formally, for a large configuration $\vec{\sigma}$ with twisted area $m$ and upper perimeter $n$, let $\vec{\sigma}'$ be the configuration such that $\sigma'_{x,y}=\uparrow$ for $(x, y)\in A$ and $\sigma'_{x,y}=\sigma_{x, y}$ for $(x, y)\not\in A$.
$\sigma'_{x,y}$ is almost a toric polyomino except that the site $(L-1, L-1)$ may not have a right or lower $\uparrow$-neighbor.
Denote the area (not twisted area!) and upper perimeter of $\vec{\sigma}'$ by $m', n'$, respectively.
Applying the bridge transformation to $\vec{\sigma}'$, we obtain $k$ ESSs with areas and upper perimeters $(m_i, n_i), i=1,2,\cdots, k$, respectively.
Assume the first ESS contains $(L-1, L-1)$.
Other ESSs must be cycle ESSs the areas are at least $L$.
In the similar fashion of lemma~\ref{lem_output_of_bridge_transformation}, we deduce:
\begin{itemize}
    \item $0\leq m'-m\leq l^2, -l^2\leq n'-n\leq l$;
    \item $m_2, m_3, \cdots, m_k\ge L$, $m_1+\cdots + m_k=m'$, $k\leq L$;
    \item $0\leq n_1+\cdots + n_k-n'\leq k$.
\end{itemize}
So 
\begin{equation}
    E_{m,n}\leq 2^{l^2}\sum_{k=1}^{L}\sum_{\substack{0\leq a\leq l^2\\-l^2\leq b\leq l}}\sum_{c=0}^{k}\sum_{\substack{m_2,\cdots,m_k\ge L,\\m_1+\cdots+m_k=m+a,\\n_1+\cdots n_k=n+b+c}}\prod_{i=1}^k(L^2D_{m_i,n_i}),
\end{equation}
Here the factor $2^{l^2}$ comes from the fact that a $\vec{\sigma}'$ may correspond to $2^{l^2}$ different configurations $\vec{\sigma}$ in $\Sigma^A$.
With the same process of~\eqref{equ_bound_on_Z_1}, we obtain
\begin{align}
    &E(x, y)=\sum_{m,n}E_{m,n}x^my^n\nonumber\\
    \leq& 2^{l^2}\sum_{k=1}^L\sum_{a=0}^{l^2}\sum_{b=-l^2}^{l}\sum_{c=0}^k x^{-a}y^{-b-c}\left(L^2\sum_{m\ge L, n}D_{m,n}x^my^n\right)^{k-1}\nonumber\\
    &\times (L^2G(x,y))\nonumber\\
    \leq& 4L^4 l^42^{l^2}x^{-l^2}y^{-l}G(x, y)\max_{k}\left(L^2\sum_{m\ge L, n}D_{m,n}x^my^n\right)^{k-1}.\label{equ_upper_bound_for_E}
\end{align}
Let $\lambda=1.2$ be a constant larger than 1, $G(\lambda^2q_a, q_p^4)$ converges to a finite value.
From the same reason of~\eqref{equ_truncation_decay_with_L}, $L^2\sum_{m\ge L, n}D_{m,n}(\lambda q_a)^m(q_p^4)^n$ is upper bounded by $c\lambda^{-L}$ for some constant $c$.
Therefore the maximum in~\eqref{equ_upper_bound_for_E} is bounded by some constant $c'$, and $E(\lambda q_a, q_p^4)$ is bounded by $L^4$ multiplied with some function $h_0(l, d, D)$.
Hence $Z_{p, \text{large}}\leq 2\sum_{m\ge L/2, n}q_a^{m}q_p^{4n}\leq 2\lambda^{-L/2}E(\lambda q_a, q_p^4)\leq 2\lambda^{-L/2}L^4 h_0(l, D, m)$.
For sufficiently large $L$, $Z_{p, \text{large}}<q_p^{4l}$, which is negligible compared with $Z_{p, \text{small}}$.

\subsection{Configurations with Small Twisted Area}
Now we focus on small configurations, i.e., configurations with twisted area less than $L/2$.
The details in $A$ of configurations become important.
In fact, we will show that the contribution of sites that are outside $A$ is small.

Let $\Sigma^{A}_0=\{\vec{\tau}\in \Sigma^{A}:\forall (x, y)\not\in A,\tau_{x,y}=\downarrow\}$ be the set of configurations that are trivial outside $A$.
For $\vec{\tau}\in \Sigma^A_0$, denote $\Sigma^A_{\tau}=\{\vec{\sigma}\in \Sigma^A_\text{small}: \sigma_{x,y}=\tau_{x,y},\forall (x, y)\in A\}$.

We now study the sites outside $A$.
All $\uparrow$ outside $A$ constitute multiple ESSs rooted at the left boundary or the top boundary.
Apply the bridge transformation to all $\uparrow$ that can arrive the left boundary via right and down movement.
The results are $k$ ESSs with area and upper perimeter $m_i, n_i$, respectively.
The bridge transformation will not increase the total upper boundary (see~\eqref{equ_perimeter_increment} in the lemma~\ref{lem_output_of_bridge_transformation}, here the configuration is small, so there will be no increment due to cycle ESSs).
Notice that the vertical boundary of these ESSs do not overlap with the boundary of $\vec{\sigma}^A$.
Therefore, these $k$ ESSs contribute an extra factor $\prod_{i}q_a^{m_i}q_p^{2n_i}$.
Adding up all possibilities of $k$, positions of roots, areas and upper perimeters of ESSs, we obtain 
\begin{align}
    &\sum_{k}\binom{l}{k}\sum_{m_i,n_i}\prod_{i}D_{m_i,n_i}q_a^{m_i}q_p^{2n_i}\nonumber\\
    =&\sum_k\binom{l}{k}G(q_a, q_p^2)^k\nonumber\\
    =&(1+G(q_a, q_p^2))^l.
\end{align}
This is the extra decay contributed by all possible ESSs rooted at the left boundary of $A$.

We do the same calculation for ESSs rooted at the upper boundary of $A$.
But at this time we should consider the horizontal boundary of these ESSs, so we need to change the priority of right and down in the step 1 of the bridge transformation (that is, connecting to the right $\uparrow$-neighbor is now prior to the down $\uparrow$-neighbor).
Other calculations are the same.
So far, what we have proved is
\begin{equation}
    \sum_{\vec{\sigma}\in \Sigma_\tau}\mathcal{B}[\vec{\sigma}]\leq (1+G(q_a, q_p^2))^{2l}\mathcal{B}[\vec{\tau}].\label{equ_outside_A_site_effect}
\end{equation}
Let $Z_{p, \text{inner}}=\sum_{\vec{\sigma}\in \Sigma^{A}_0}\mathcal{B}[\vec{\sigma}]$.
Summing~\eqref{equ_outside_A_site_effect} over all $\vec{\tau}\in \Sigma^A_0$, we obtain that $Z_{p, \text{small}}\leq (1+G(q_a, q_p^2))^{2l}Z_{p,\text{inner}}$.
Therefore, we only need to bound the contributions of configurations that are trivial outside $A$.

\subsection{Upper Bound of Theorem~\ref{theorem:entropy}}
Now we are ready to prove the upper bound of Theorem~\ref{theorem:entropy}.
By~\eqref{equ_outside_A_site_effect}, we only need to consider $\vec{\tau}\in \Sigma^A_0$.
Let $E_{m,n}^0$ be the number of configurations in $\Sigma_{0}^A$ with area $m$ and upper perimeter $n$.
By the reflection trick, $Z_{p, \text{inner}}\leq 2\sum_{m,n}E_{m,n}^0q_a^mq_p^{4n}$.

Let $\vec{\sigma}\in \Sigma^A_0$ be a configuration with twisted area $m$ and upper perimeter $n$.
Since $\uparrow$ can appear in at most $n$ columns and each columns has at most $l$ $\uparrow$, we obtain a simple relation $l^2\ge m\ge l^2-nl$.

For $(x, y)\in A$ and $\sigma_{x, y}=\downarrow$, one of $\sigma_{x+1, y}, \sigma_{x, y+1}$ must be $\downarrow$ (otherwise $f(-\sigma_{x, y}, -\sigma_{x, y+1}, -\sigma_{x+1, y})=0$).
All $\downarrow$ in $A$ constitute several ESSs rooted at the right and bottom boundary of $A$.
We extend the region $A$ by one site towards the bottom and right boundary of $A$.
Specifically, suppose $A=[a, a+l-1]_\ZZ\times [b, b+l-1]_\ZZ$, where $I_\ZZ:=I\cap \ZZ$ for an interval $I$.
The extended region is $A^{\text{ext}}:=[a, a+l]_\ZZ\times [b, b+l]_\ZZ$.
Then all $\downarrow$ in $A^{\text{ext}}$ constitute exactly one ESS rooted at $(a+l, b+l)$.
The ESS has area $m+2l+1$ and upper perimeter $n'$ and can be directly converted to a plain polyomino via step 4 of the bridge transformation.

We claim that that $n'=n+a$, where $a\in [1, l+1]_\ZZ$ is the number of $\downarrow$ in the first row of $A^{ext}$.
The vertical perimeter of $\vec{\tau}$ and the ESS is $2n, 2n'$, respectively.
Then,
\begin{itemize}
    \item The lower boundary of $\downarrow$ in the last row of $A^{ext}$ is counted in $2n'$ but not in $2n$;
    \item The upper boundary of $\downarrow$ in the first row is counted in $2n'$ but not in $2n$;
    \item The upper boundary of $\uparrow$ in the first row is counted in $2n$ but not in $2n'$.
\end{itemize}
So $2n'-2n = (l+1)+a-(l+1-a)=2a$ and $n'=n+a$.
Therefore, $E_{m,n}^0\leq \sum_{a=1}^{l+1}D_{m+2l+1, n+a}$ and
\begin{align}
    Z_{p,\text{inner}}\leq 2\sum_{m\ge l^2-nl}\sum_{a=1}^{l+1}D_{m+2l+1, n+a}q_a^m q_p^{4n}.\label{equ_upper_bound_of_Z_p_0}
\end{align}
The generating function $G(x, y)$ gives us ample information about $D_{m,n}$:
\begin{lemma}
    There exists a constant $c$ such that for any $m, n\ge 1$, $D_{m,n}\leq c(1.9)^{m+2n}$
\end{lemma}
\begin{proof}
    Regard $G(t, t^2)=\frac{t}{2}(\sqrt{\frac{(1+t)(1+t-t^3)}{(1-2t+t^2-t^3-t^4)}}-1)=\sum_{m,n}D_{m,n}t^{m+2n}$ as a power series of $t$.
    The coefficient of $t^k$ is $g_k=\sum_{m+2n=k}D_{m,n}$.
    The radius of convergence is $r\approx 0.5271$ (the positive real zero root of $1-2t+t^2-t^3-t^4$).
    A well-known result in analysis states that, $\lim\sup_{k\to\infty} |g_k|^{1/k}=\frac{1}{r}<1.898$.
    So there exists a constant $c$ such that $g_k\leq c(1.9)^k$.
    Hence $D_{m,n}\leq g_{m+2n}\leq c(1.9)^{m+2n}$.
\end{proof}
Now we can easily bound the right hand side of~\eqref{equ_upper_bound_of_Z_p_0}:
\begin{align}
    &Z_{p, \text{inner}}\nonumber\\
    \leq& 2c\sum_{l^2\ge m\ge l^2-nl}\sum_{a=1}^{l+1}(1.9)^{m+2n+2l+1+2a}q_a^mq_p^{4n}\nonumber\\
    =&c'(1.9)^{4l}\sum_{l^2\ge m\ge l^2-nl}(1.9q_a)^m(3.61q_p^4)^n\nonumber\\
    =&c'(1.9)^{4l}\sum_{n\ge l, m\ge 0}(1.9q_a)^m(3.61q_p^4)^n\nonumber\\
    &\quad+c'(1.9)^{4l}\sum_{n=0}^l\sum_{m=l^2-nl}^{l^2}(1.9q_a)^m(3.61q_p^4)^n\nonumber\\
    \leq& c''(1.9)^{4l}\left((3.61q_p^4)^{l}+\sum_{n=0}^{l}(1.9q_a)^{l^2-nl}(3.61q_p^4)^n\right)\label{equ_upper_bound_for_thm_3}\\
    \leq& c''(1.9)^{4l}\left((3.61q_p^4)^{l}+(l+1)\max\{(1.9q_a)^{l^2}, (3.61q_p^4)^{l}\}\right). \label{equ_upper_bound_for_p_inner}
\end{align}
Assume $d\ge 2, D\ge 3$, then $q_a\leq 1/2, q_p\leq 1/3$ and $1.9q_a<1, 3.61q_p^4<1$, so all geometric series converges. $c, c', c''$ absorb all constants.
Therefore, for some universal constant $c$ (different from the previous $c$),
\begin{align}
    Z_p&=Z_{p, \text{large}}+Z_{p, \text{small}}\nonumber\\
    &\leq q_p^{4l}+(1+G(q_a, q_p^2))^{2l}Z_{p, \text{inner}}\nonumber\\
    &\leq c(1.91q_a)^{l^2}+c(2.9q_p)^{4l},
\end{align}
where $G(q_a, q_p^2)\leq G(1/2, 1/9)<0.145$. For this upper bound to be meaningful, the right hand side should be less than 1.
Take $l$ such that $(1.91/2)^{l^2}, (2.9/3)^{4l}<\frac{1}{2c}$.
A sufficient condition is $l\ge \log_{1.05}(2c)$, i.e., $l$ is larger than a constant.

\subsection{Lower Bound of Theorem~\ref{theorem:entropy}}
By slightly modifying $\vec{\sigma}_0, \vec{\sigma}_1$, we construct a family of configurations with similar $\mathcal{B}$ value, which provides a lower bound of $Z_p$.
Here we still assume $A$ is the region $[L-l, L)_\mathbb{Z}\times [L-l, L)_\mathbb{Z}$.
A caveat is that Eq.~\eqref{equ_entanglement_single_site} is only an upper bound instead of an equation when $(x,y,z)=(\downarrow, \uparrow, \uparrow)$, $(\uparrow, \downarrow, \uparrow)$, or $(\uparrow, \uparrow, \downarrow)$.
To fill this gap, we define $\gamma_1=f(\downarrow,\uparrow, \uparrow)/q_p^2, \gamma_2=f(\uparrow, \downarrow, \uparrow)/(q_aq_p)=f(\uparrow, \uparrow, \downarrow)/(q_aq_p)$.
Note that $\gamma_1,\gamma_2$ are less than but close to 1.
In fact, one can prove that $\gamma_1,\gamma_2\ge 0.74.$

\subsubsection*{Modifying $\vec{\sigma}_0$}
By definition, in $\vec{\sigma}_0$ all spins inside $A$ are $\uparrow$ and all spins outside $A$ are $\downarrow$.
$\mathcal{B}[\vec{\sigma}_0]=f(\downarrow, \downarrow, \uparrow)^{4l-2}f(\downarrow, \uparrow, \uparrow)= q_p^{4l}\gamma_1$.
We flip $k$ spins in the boundary.
The twisted area becomes $k$ and the perimeter is at most $4l+2k$.
The number of $\gamma_1, \gamma_2$ are at most $k+1, k$, respectively.
Hence we obtain one part of the lower bound:
\begin{align}
    \sum_{k=0}^{4l-1}\binom{4l-1}{k}q_a^kq_p^{4l+2k}\gamma_1^k\gamma_2^{k+1}=q_p^{4l}\gamma_2(1+q_aq_p^2\gamma_1\gamma_2)^{4l-1}.
\end{align}

\subsubsection*{Modifying $\vec{\sigma}_1$}
In $\vec{\sigma}_1$ all spins are $\downarrow$.
$\mathcal{B}[\vec{\sigma}_0]=f(\uparrow, \uparrow, \uparrow)^{l^2}=q_a^{l^2}$.
We flip $k$ spins inside $A$.
To make the $\mathcal{B}$-value of the resulting configuration nonzero, we require that the $k$ spins are in the odd row of $A$ but not in the bottom or right boundary of $A$. 
The number of such positions is $l(l-1)/2$ for even $l$ and $(l-1)^2/2$ for odd $l$.
The twisted area becomes $l^2-k$ and the perimeter is at most $4k$.
The number of $\gamma_1, \gamma_2$ are at most $k$.

So we obtain the other part of the lower bound:
\begin{align}
    \sum_{k=0}^{(l-1)[l/2]}\binom{(l-1)[l/2]}{k}q_a^{l^2-k}q_p^{4k}\gamma_1^k\gamma_2^k=q_a^{l^2}(1+q_a^{-1}q_p^4\gamma_1\gamma_2)^{(l-1)[l/2]}.
\end{align}
To avoid the annoying floor function, we use the inequality $(l-1)[l/2]\ge (l-1)^2/2$.
Substituting $\gamma_2\ge 0.7$ and $\gamma_1\gamma_2\ge 0.5$, we obtain the lower bound of Theorem~\ref{theorem:entropy}.

\subsection{Proof of Theorem~\ref{theorem:maximum_entropy}}
\label{theorem3_proof}
In this section we prove that $Z_p\leq q_a^{l^2}+O(q_p^4)$ if $D\ge 3$,  indicating that the entanglement entropy is close to maximum entropy for large bond dimensions.

We only need to modify the proof of Theorem~\ref{theorem:entropy} slightly.
Notice that in $\Sigma^A_0$, there is only one configuration with $n=0$, i.e., the typical configuration $\vec{\sigma}_1$.

For other configurations, $n\ge 1$, so~\eqref{equ_upper_bound_for_thm_3} becomes
\begin{align}
    &c''(1.9)^{4l}\left((3.61q_p^4)^{l}+\sum_{n=1}^{l}(1.9q_a)^{l^2-nl}(3.61q_p^4)^n\right)\nonumber\\
    \leq& c''(1.9)^{4l}\left((3.61q_p^4)^{l}+l\max\{(1.9q_a)^{l^2-l}(3.61q_p^4), (3.61q_p^4)^{l}\}\right).
\end{align}
Multiplying $(1+G(q_a, q_p^2))^{2l}<(1.45)^{2l}$, we obtain the upper bound of contributions of all small configurations that have at least one $\uparrow$ in $A$: $c(1.91q_a)^{l^2-l}(3.61q_p^4)+c(2.9q_p)^{4l}=O(q_p^4)$.

The remaining part is small configurations that have no $\uparrow$ in A, i.e., configurations in $\Sigma^A_{\vec{\sigma}_1}$.
Suppose $\vec{\sigma}\in \Sigma^A_{\vec{\sigma}_1}$ with $m$ $\uparrow$ outside $A$ (thus the twisted area is $m+l^2$) and suppose the upper perimeter is $n$
All $\uparrow$ outside $A$ constitute several (at most $2l$) ESSs rooted at the left and top boundary.
Use the standard bridge transformation to decompose them and convert them into plane polyominoes with area and upper perimeter $(m_i,n_i)_{1\leq i\leq k}$, respectively.
With the reflection trick, we obtain an upper bound
\begin{align}
    \sum_{\vec{\sigma}\in \Sigma^A_{\vec{\sigma}_1}}\mathcal{B}[\vec{\sigma}]&\leq q_a^{l^2}\sum_{k=0}^{2l}\binom{2l}{k}\sum_{m_i,n_i\ge 0}\prod_{i=1}^{k}(D_{m_i, n_i}q_a^{m_i}q_p^{n_i})\nonumber\\
    &= q_a^{l^2}(1+G(q_a, q_p^4))^{2l}.
\end{align}
Here $G(q_a, y)\leq G(1/2, y)=\frac{y}{2}(\sqrt{\frac{9-3y}{1-3y}}-1)\leq 5y$ for $y=q_p^4\leq 1/81$.
So
\begin{align}
    \sum_{\vec{\sigma}\in \Sigma^A_{\vec{\sigma}_1}}\mathcal{B}[\vec{\sigma}]&\leq q_a^{l^2}(1+G(q_a, q_p^4))^{2l}\nonumber\\
    &\leq q_a^{l^2}(1+5q_p^4)^{2l}\nonumber\\
    &=q_a^{l^2}+q_a^{l^2}q_p^4\sum_{i=1}^{2l}\binom{2l}{i}5^iq_p^{4(i-1)}\nonumber\\
    &<q_a^{l^2}+q_a^{l^2}q_p^4\sum_{i=1}^{2l}\binom{2l}{i}5^{2l}\nonumber\\
    &<q_a^{l^2}+q_p^4 q_a^{l^2}10^{2l}=q_a^{l^2}+O(q_p^4).
\end{align}

In summary, $Z_{p}\leq Z_{p, \text{large}}+Z_{p, \text{small}}\leq q_p^{4l}+O(q_p^4)+q_a^{l^2}+O(q_p^4)=q_a^{l^2}+O(q_p^{4})\leq d^{-l^2}+O(D^{-4}).$
Obviously $Z_{p}\ge d^{-l^2}$ since the right hand side is the maximum entropy.
So $d^{-l^2}\leq E[\text{Tr}(\rho_A^2)]\leq d^{-l^2}+O(D^{-4})$.
By Markov inequality, $\text{Pr}\left[ \left|\text{Tr}(\rho_A^2)-d^{-|A|} \right| \geq \epsilon \right] \leq O(\frac{D^{-4}}{\epsilon})$.

\bibliography{MPS, reference, reference_Deng}

\end{document}